\documentclass[acmsmall]{acmart}

\usepackage{booktabs}   
\usepackage{subcaption} 

\usepackage{amsmath}
\usepackage{mathpartir}
\usepackage{proof}
\usepackage{tikz}
\usepackage{dutchcal}
\usepackage[capitalise]{cleveref}
\usepackage{caption}
\usepackage{subcaption}

\usepackage{bbm}
\usepackage{wasysym}

\newcommand{\tick}[1]{\texttt{tick}\{#1\}}

\newcommand{\fun}[3]{\texttt{fun}\{#1\}\; #2 . #3}
\newcommand{\app}[2]{#1\; #2}

\newcommand{\node}[3]{\texttt{Node}(#1;#2;#3)}
\newcommand{\leaf}[0]{\texttt{Leaf}}

\newcommand{\matcht}[6]{\texttt{match}\;#1\;\texttt{with}\;\texttt{Leaf}\; \rightarrow #2 \mid \texttt{Node}(#3,#4,#5) \; \rightarrow \; #6}
\newcommand{\elet}[3]{\texttt{let}\;#1\;=\;#2\;\texttt{in}\;#3}
\newcommand{\ite}[3]{\texttt{if}\;#1\;\texttt{then}\;#2\;\texttt{else}\;#3} 
\newcommand{\closure}[4]{\texttt{C}_{#1}(#2; #3.#4)}

\newcommand{\pot}[2]{\langle #1; #2 \rangle}

\newcommand{\basic}{\mathbf{basic}}
\newcommand{\unit}{\mathbf{unit}}

\newcommand{\tbool}{\mathbf{bool}}

\newcommand{\sub}{<:}
\newcommand{\wv}{\mathsf{wv}}
\newcommand{\li}{\mathsf{loc}}

\crefname{section}{\S\!}{\S\S\!}
\crefname{subsection}{\S\!}{\S\S\!}

\crefformat{equation}{Ineq. \!\!(#2#1#3)}

\newtheorem{theorem}{Theorem}[section]

\newtheorem{lemma}[theorem]{Lemma}

\AtBeginDocument{%
  \providecommand\BibTeX{{%
    \normalfont B\kern-0.5em{\scshape i\kern-0.25em b}\kern-0.8em\TeX}}}

\begin{document}

\title{Automatic Amortized Resource Analysis with the Quantum Physicist's Method}

\author{David M Kahn}
\email{davidkah@andrew.cmu.edu}
\author{Jan Hoffmann}
\affiliation{%
  \institution{Carnegie Mellon University}
  \streetaddress{5000 Forbes Ave}
  \city{Pittsburgh}
  \state{Pennsylvania}
  \country{USA}
  \postcode{15232}
}

\renewcommand{\shortauthors}{Kahn and Hoffmann}

\begin{abstract}
  We present a novel method for working with the physicist's method of amortized resource analysis, which we call the \textit{quantum physicist's method}. These principles allow for more precise analyses of resources that are not monotonically consumed, like stack. This method takes its name from its two major features, \textit{worldviews} and \textit{resource tunneling}, which behave analogously to quantum superposition and quantum tunneling. We use the quantum physicist's method to extend the Automatic Amortized Resource Analysis (AARA) type system, enabling the derivation of resource bounds based on tree depth. In doing so, we also introduce \textit{remainder contexts}, which aid bookkeeping in linear type systems. We then evaluate this new type system's performance by bounding stack use of functions in the Set module of OCaml's standard library. Compared to state-of-the-art implementations of AARA, our new system derives tighter bounds with only moderate overhead.
\end{abstract}



\maketitle

  \section{Introduction}
\label{sec: intro}

%

Resource analysis of programs is an ongoing area of study. Techniques range from recurrence solving \cite{POPL:KML20,kincaid2017compositional}, to term re-writing \cite{RTA:AM13,naaf2017complexity}, to type systems \cite{POPL:CW00,dal2013geometry}, and beyond. Because automatic resource analysis is generally undecidable, it is a challenge in such work to balance analysis strength with programmer burden. Some opt for more analytical power at the expense of requiring more manual intervention \cite{nipkow2019amortized,gueneau2018fistful}. Others emphasize keeping such analyses automatable \cite{albert2007cost,gulwani2009speed}, so that their use remains accessible to non-experts. The ideas of the \textit{quantum physicist's method}, which we introduce in this work, fall into the latter camp. The new method allows us to obtain tighter analyses without hurting automatability.

%

Usually frameworks for automated resource analysis focus on resources that are monotonely consumed, like time. Only a few techniques, like \cite{albert2015non,vasconcelos2008space,campbell2009amortised,hofmann2003static}, have been able to extend to \textit{non}-monotonely consumed resources, like space. For these resources, the \textit{high-water mark} of resource use is the pertinent metric, rather than the net cost. Even for techniques that can reason about the high-water mark, performing tight analyses remains challenging.  This is one problem our work addresses for techniques that use the \textit{physicist's method} of \cite{kn:Tarjan85} for resource analysis. 


Automatic Amortized Resource Analysis (AARA) is a type-based technique
that uses the physicist's
method~\cite{hofmann2003static,hoffmann2017towards}. This system
endows its types with a notion of \textit{potential energy}. As values
of these types are created or destroyed in computation, they take in
or release this potential to pay for costs. This is how the
physicist's method tracks resources, just like a physicist tracks the
energy change in a physical system. To infer these types, AARA reduces
the synthesis of valid potential annotations to linear program
solving, allowing for efficient automation. Then when AARA types a
function, that type's potential certifies the resource use of that
function: The potential of the argument bounds the high-water mark
cost, and the difference between the potential of the argument and
return types bounds the net cost.


While AARA can reason about the high-water mark, the bounds it derives can be quite loose. This occurs for two reasons, both stemming from how AARA localizes potential:



\begin{enumerate}

\item There is no way of knowing ahead of time which way a control flow branch will resolve. Thus, if each branch draws its cost out of different sets of values, AARA must ready potential on both of those sets. Potential localized onto the unused values will then be wasted. This can double the potential AARA requires to type a branching expression, and affects both monotone and non-monotone resources.

\item When non-monotone resources are returned, reusing them effectively is critical to achieving a good resource bound. However, AARA is only able to store resources on values during their creation, which may not coincide with the return of resources. This can mean that all returned resources are lost. Each time such code is run in sequence, this effect compounds.

\end{enumerate}

Indeed, these problems are not constrained to AARA, but are present in every existing system that uses the physicist's method. This dates back to the seminal paper on the technique \cite{kn:Tarjan85}, where it was suggested that potential should be assigned locally to the program's data structures.





To solve the first of these problems, the quantum physicist's method introduces \textit{worldviews}. Worldviews fix the branching problem by allowing for a lazy potential assignment to values. This means potential does not need to be readied until the branch has resolved, so potential is not wasted for the unused branch. Worldviews work analogously to quantum superposition by maintaining a collection of different potential allocations, which we only resolve as needed.






To solve the second problem, we want to be able to assign resources on a temporary basis, so that they can be reassigned later. For this purpose, the quantum physicist's method introduces \textit{resource tunneling}. Resource tunneling uses worldviews in the same way quantum tunneling uses superposition. The basic principle is that, if in \textit{some} worldview potential could be allocated to meet the high-water mark, then \textit{all} worldviews are justified in paying only the net cost. This reasoning principle turns out to be sound even if some worldviews might temporarily track inconsistent potential allocations, with negative potential assigned to some values.





We implement the quantum physicist's method ideas in a functional language by extending the AARA type system. In this setting, worldviews are a set of simultaneous type derivations for the same expression. To further include the principle of resource tunneling, we also extend the type system with novel \textit{remainder contexts}. These remainder contexts track the potential remaining on values after an expression's evaluation, so that resource tunneling may reassign it. Remainder contexts turn out to be useful in their own right, as their lazy tracking of remaining potential ensure that values waste as little of their potential as possible. This functionality actually subsumes the previous \textit{sharing} construct of AARA.





By conservatively extending AARA with our new reasoning principles, we maintain its benefits while giving tighter bounds for non-monotonic resources. The system supports amortization, type inference is still reducible to linear program solving, and type derivations still certify resource bounds. The system can still handle programmer-defined cost models for general resources, and can compose its analyses together. And the system is still backed by a soundness theorem proven with respect to an operational cost semantics. However, unlike AARA, our new type system can better analyze resources like stack, and can even use its new features to base potential on tree depth. This enables tight automated analysis of resource usage like the stack space of binary searches or tree traversals.






We implemented a prototype of our new type system in OCaml, and test its efficacy against Resource Aware ML (RaML), a state-of-the-art implementation of AARA. As a case study, we use them to analyze the stack space of functions from the OCaml standard library Set module, which implements sets with binary trees. Experiments show our prototype attains significantly better resource bounds with only moderate overhead. Our prototype can derive bounds for 36 out of 37 functions tested, and obtains tighter results than RaML on 30 cases. RaML never obtains a tighter result than our implementation.


To summarize our contributions:
\begin{itemize}
	\item We introduce worldviews and resource tunneling, new quantum-inspired reasoning principles to adapt local potential functions to nonlocal phenomena for use in the physicist's method of amortized analysis.
	
	\item We design a type system to incorporate these new principles into AARA using remainder contexts and simultaneous typings. As a result, the system can express a wider range of bounds, such as functions of tree depth .

        \item We prove the soundness of the new type system with respect to an operational cost semantics that is parametric in the resource of interest.
	
	\item We provide a prototype implementation and experimentally show that it can derive tighter bounds with moderate performance overhead.
\end{itemize}

\section{Overview}

In this section,we first describe the state-of-the-art in AARA and
then informally explain our two main technical innovations: worldviews
and resource tunneling.

\subsection{Automatic Amortized Resource Analysis}
\label{subsec: aara}

AARA uses types to perform the physicist's method of amortized analysis.
In amortized analysis~\cite{kn:Tarjan85}, bursty costs are smoothed
out through incremental prepayment, or \textit{amortization}.
The physicist's method is a way of accounting for amortization by
imbuing program state with a notion of potential energy.
To enable automatic inference, this potential assignment is locally
defined over data structures like lists and trees in AARA.
The key invariants of this accounting is that the difference in
potential between any two points of execution is enough to pay for the
net cost, and the initial potential is sufficient to cover the
high-water mark of cost. 
The type AARA assigns a function captures these invariants: the
argument type's potential certifies an upper bound on the high-water
mark, and the difference between the argument and the return type's
potentials certifies an upper bound on the net cost.

AARA is parametric in the resource of interest and can handle a
variety of language features and potential functions. Aside from time, 
the system has also been used to measure more complicated resources like heap 
\cite{hofmann2003static} and smart contract gas usage \cite{das2021resource}.
For functional languages, AARA-based systems like
RaML~\cite{hoffmann2017towards} support user-defined
datatypes~\cite{Jost09}, higher-order functions~\cite{Jost10}, and
polynomial potential functions~\cite{HoffmannAH12}.
AARA can also be applied to imperative
languages~\cite{atkey2010amortised,CarbonneauxHZ15} and to derive bounds on
the expected cost of probabilistic
programs~\cite{NgoCH17,wang2020raising}.
Even with all of these features in play, bound inference can be
reduced to solving a system of linear inequalities, allowing efficient
resource analysis with off-the-shelf linear program solvers.
Moreover, a soundness theorem proves that the derived bounds are sound
with respect to an operational cost semantics and type derivations are
easily-checkable bound certificates.

To simplify the presentation of this work, we build our type system off of a fragment of AARA. As presented, our type system supports only linear resource functions (linear AARA), and excludes types like products and sums. Nonetheless, this work can be extended to all of the features of AARA, and still reduces to linear constraint solving.
To illustrate the analysis in this section, we assume a cost model
that counts the number of stack frames allocated during an evaluation.
For simplicity, we assume that every function call results in the
allocation (before the call) and deallocation (after the call) of a
stack frame.







\paragraph{Example}
To see how AARA expresses potential, consider the \verb"map" function
from the OCaml List module. This is a higher order function with the
type
$$\verb"map": (\tau \rightarrow \rho) \rightarrow L(\tau) \rightarrow L(\rho) \; .$$
%
To include potential, suppose \verb"map" is given a function \verb"f"
that requires 1 stack frame and returns 1 stack frame. 
This models a function $f$
that performs single auxiliary function call in its body.
Then the evaluation of the expression \verb"map f l" requires $n + 2$
stack frames where $n$
is the length of the list $l$
and $n + 2$ stack frames become available again after the evaluation.

The function \verb"map f" should be typed to take a list with $1$
unit of resource per element and $2$
additional units, and return a list with $1$
per element and $2$
additional units.
This can be expressed with a potential-annotated
type like
$$\verb"map" : \pot { \pot \tau 1 \rightarrow \pot \rho 1 } 0 \rightarrow { \pot
  {L^1(\tau)} 2 \rightarrow \pot {L^1(\rho)} 2 }$$
where the angle
brackets pair types with potential, and the superscript on lists gives
the potential per element.%
\footnote{ In this work, it is beneficial to separate potential
  annotations from types and combine them in an annotation context
  $P$.
  The type of the function map then has the form
  $\langle (\tau \rightarrow \rho) \rightarrow L(\tau) \rightarrow
  L(\rho) ; P \rangle$. This notation will be discussed formally in
  \cref{sec: type}.}
Examining the type of \verb"map f", we see that the difference
$2 + n - (2 + m)$
between the input and output potential is an upper bound on the net
cost of running the function, where $n$
is the length of the input list and $m$ is the length of the result.
Since $n=m$,
the net cost is $0$
and the bound on the high-water mark is given by the initial potential
$n + 2$.

The AARA approach of imbuing types with potential allows for greater
compositionality.
Consider the function
\begin{center}
\verb"let map2 f l = map f (map f l)".  
\end{center}
We can assign \verb"map2" the same type that we assigned \verb"map",
providing the correct high-water mark bound $n + 2$.
This bound is justified by the type system since the expression
\verb"map f l" of type list has the type $\pot
  {L^1(\tau)} 2$, which matches the type of the second argument
of \verb"map".

The function \verb"map" can be given many different potential
annotations, specialized depending on its arguments and us of the
result in the continuation.
For example, we would use the type
$$\verb"map" : \pot { \pot \tau 1 \rightarrow \pot \rho 1 } 0 \rightarrow { \pot
  {L^2(\tau)} 2 \rightarrow \pot {L^2(\rho)} 2 }$$
if the result of a call of \verb"map" is used as the argument of
function that requires the type $L^2(\tau)$.
In general, type annotations of \verb"map" can be described with
linear inequalities that state e.g. the potential of the
input list $q$ is at least the potential $p$ of the resulting list.

\paragraph{Limitations}

Not all code is tightly analyzed by AARA.
Sometimes this is for meta-theoretical reasons like the Halting
Problem, or because the code behaves based on semantic rather than
structural properties.
However, even common code patterns can cause problems for non-monotone
resources like stack frames.  
Consider the following expression, where the function $g$ is given the same type as $f$ from above.
\begin{center}
  \verb"let x = map f l in let y = map g l in ()"
\end{center}
The evaluation of the expression needs $n + 2$
stack frames where $n$
is the length of the list $l$.
However, AARA can only derive the loose bound $2n + 2$.
The problem is that the potential on the two uses of $l$
has to be summed up.
The potential that becomes available after the first call of
\verb"map" is assigned to $x$ but cannot be transferred back to $l$.

Other problematic examples are tree traversals as implemented in
\cref{fig: excode}.
The function \verb"bin" implements a binary search and the function
\verb"size" is a tree traversal that computes the number of inner
nodes.
Both functions need $h$
stack frames in their evaluation, where $h$
is the height of the tree $t$ in the argument.
However, AARA can only derive a bound based on the number of nodes in
the tree, rather than its depth.
Such a bound can be exponentially loose.
Code like the binary search from \cref{fig: bin} poses a problem
because the analysis does not know which subtree will be searched. 
And when bounding the stack usage of a tree traversal like \cref{fig: tsize}, AARA does
not recognize that the stack from the left subtree can be reused for
the right. Because AARA must eagerly assign potential localized
to a given subtree, it readies potential for \textit{both} subtrees in
each of these two examples.

\begin{figure}
\begin{subfigure}{0.49\textwidth}
\centering
\begin{verbatim}
let rec bin x t = match t with 
  | Leaf -> false
  | Node(v,t1,t2) -> if x = v then true
    else if x < v then bin x t1
    else bin x t2
\end{verbatim}
\caption{Binary Search}
\label{fig: bin}
\end{subfigure}
\begin{subfigure}{0.49\textwidth}
\centering
\begin{verbatim}
let rec size t = match t with 
  | Leaf -> 0
  | Node(_,t1,t2) -> size t1 + size t2 + 1
\end{verbatim}
\caption{Tree Size}
\label{fig: tsize}
\end{subfigure}
\caption{Example Code}
\label{fig: excode}
\end{figure}

\paragraph{Quantum Motivation}

There is historical precedent for the value of the comparison between resource analysis and physics: In Tarjan and Sleator's description of amortized analysis \cite{kn:Tarjan85}, they supply a similar metaphor based on conservation of energy, and the resulting \textit{physicist's method} has been of pedagogical value since then. This view has also proven by fruitful in application, as witnessed by systems like AARA and others \cite{cutler2020,gueneau2018fistful,timecredits,atkey2010amortised}. It is therefore natural to consider further interdisciplinary parallels between physical reasoning and resource analysis.  

The concepts presented in this work represent a refinement of the physicist's method framework that qualitatively parallels the refinement of general physical principles to quantum principles. The very development of the system was guided by quantum physical intuition: \textit{resource tunneling} only came about when the parallels between \textit{worldviews} and quantum superposition were realized. And these connections are not a fluke: both quantum phenomena and resource usage follow linear logical rules, so they should be expected that they have similar behaviours. Indeed, many additional parallels will be pointed out throughout this work.

\subsection{Worldviews}

Our first innovation are worldviews, which can be thought of as simultaneous type derivations for the same expression. When an expression might have potential assigned in multiple ways, different worldviews will cover each of the assignments. This allows the choice of assignment to be kept in suspense, and determined lazily when needed. This capability can be used to solve the problems faced in analyzing branching code, because differing solutions may specialized for each branch. 

Type-theoretically, worldviews can be thought of similarly to introducing a top-level additive sum type across the product types describing a context. The additive sum ($\&$) is a linear type connective that represents being able to specialize to either type, similarly to type intersection. It does not distribute over products, so that in general $(x:\alpha_0 \otimes y:\beta_0) \&  (x:\alpha_1 \otimes y:\beta_1) \neq x:(\alpha_0 \& \alpha_1) \otimes y:(\beta_0 \& \beta_1)$. This non-distributivity makes the specialization of a value's additive sum become dependent on the specialization of other values, introducing non-local effects on typing. This is necessary in a resource-aware setting, as one cannot both have their cake and eat it too: If we split our cake between our plate $x$ and our stomach $y$, distributivity would allow us to choose the best assignment to each independently. Despite the similarity to additive sums, it should be noted that worldviews are not identical to them, due to mechanics introduced by resource tunneling.

\paragraph{The Issue with AARA} Let us look at the case of binary search (\cref{fig: bin}) to see where the current AARA-style analysis fails. To tightly analyze the number of recursive calls, we want to assign potential based equal to its depth. This potential (less 1 to pay cost) should be divided among the two subtrees to properly type them for recursive calls. There are at least two potential assignments to consider: giving the left subtree all the depth potential, or giving the right subtree all of it. The correct choice will depend upon which branch is taken by the search. However, the branch is determined \textit{after} the point where the subtrees are created, and is not statically determined. This precludes any AARA-style static localization of potential to the subtrees.




\paragraph{Our Solution} To solve this problem, we want to be able to hold the decision of how to split potential in suspense, and determine it later when the branch occurs. This introduces the physicist method's potential to a quantum physics phenomenon: quantum superposition. In physics, various quantities like mass and energy are conserved, just like potential is. When a particle splits into two, the new particles may enter a superposition of states, each of which represents a different division of these conserved quantities. One interpretation of this superposed existence is as branching worlds \cite{everett1957relative}. Only upon observation is the particular state determined, \textit{collapsing} the superposition. In quantum physics, this collapse is probabilistic, but here we may deterministically choose the most useful way for our type derivation.

Superposition is actually not such a mysterious way to bookkeep for classical resources. Consider that Alice and Bob are given \$5 by their parents to spend at a candy store. On their way to the store, they might split the money 50/50, or have Alice carry it all, or use any other split they desire. Only when the two finally spend the money at checkout is it determined how much each must carry. If we were keeping a real-time account of their finances, it would not suffice to assign any one split to the \$5 on their way to the store. Instead, the money is better tracked as if it were in superposition. Then at checkout we could collapse this superposition into the proper split between Alice and Bob.

\paragraph{Example } Now let us apply worldviews to binary search. The number of recursive calls should be equal to the depth of the input tree, so take the input to have 1 unit of depth potential. When pattern matching the input tree, we get 1 unit of potential to pay for a recursive call, and create two subtrees. Instead of assigning one subtree all the depth potential, we assign it to both, each in a different worldview. This leaves the worldviews tracking a superposition of two different, but equally valid, type derivations for the pair of subtrees. Then we can collapse to the left assignment when analyzing the left search branch, and vice versa for the right. This effectively delays the choice of the potential assignment until after the branch. In each branch, we then find that the subtree is properly typed with 1 unit of depth potential for the recursive call. This allows the analysis to find the desired depth bound.

\paragraph{Type System}
In the remainder of this paper we develop a type system that supports
worldviews while preserving the benefits of AARA such as type
inference with LP solving, a formal soundness proof with respect to a
cost-semantics as well as natural compositionality.
For the latter, it is important that worldviews can also be used to
construct trees with potential assigned to their height without waste.
The key idea is that we require the existence of two worldviews: One
that assigns potential to the left subtree and one that assigns
potential to the right subtree.
Also note that two worldviews in the type derivation of the function
\verb"bin" suffice to cover all possible paths that the binary search
can take when traversing a tree.

\subsection{Resource Tunneling}
\label{sec: rt}

Our second innovation is resource tunneling; a reasoning principle that allows us to only temporarily assign potential to some location, and change the assignment later. This means that there does not need to be a single consistent history of potential assignments, and instead we can jump between those most beneficial at a given time. This principle allows better analysis of sequenced operations on non-monotonic resources, like tree traversals' stack bounds, because subsequent operations can make use of potential assignments that would be inconsistent with previous assignments.

Type theoretically, this principle allows our worldview additive sum types to "undo" their specialization under certain conditions. This capability is somewhat more dynamic than the usual treatment of additive sums. It does this by allowing otherwise-unsound specializations, with local instances of negative potential, for temporary bookkeeping of resources. The specific mechanics of this bookkeeping are described in the following paragraphs.
 
\paragraph{The Issue with AARA}
Let us look at the function \verb"size" (\cref{fig: tsize}) to see where the current AARA-style analysis for stack bounds derives loose bounds. To tightly analyze the number of stack frames, we want to assign potential based equal to its depth. This potential (less 1 to pay cost) should be divided among the two subtrees to properly type them for recursive calls. However, no division of potential is sufficient. If the left subtree gets all of the potential to be typed correctly for its recursive call, the right subtree's potential would go negative during its recursive call, and vice versa. Even worldviews do not solve this problem. What is needed is a justification for reassigning the returned stack frames from the left subtree traversal to the right subtree. 

\paragraph{Our Solution}
To solve this problem, we allow for the temporary negativity of some potential in some worldviews, so long as other worldviews properly justify it with resource tunneling: So long as \textit{some} worldview can justify the high-water mark of cost, then \textit{all} worldviews can only consider net cost. This can be justified by interpreting the witness worldview as showing how to reallocate resources for temporary lending in such a way that no potential goes negative. 

Resource tunneling is analogous to another quantum phenomenon: quantum tunneling. This phenomenon allows low-energy particles to pass high-energy potential barriers. The way this works is due to the Heisenberg uncertainty principle: the more well determined a particle's location is, the less its energy can be \cite{griffiths2018introduction}. Thus if a particle is well-known to be prior to the barrier, its energy must be less well known, and is treated as a distribution over superposed states. While by and large the particle's states may seem low-energy, a very small fraction of those states have high enough energy to pass the potential barrier. Since the particle behaves somewhat as if it was in any of those states, it can sometimes pass over the barrier. Crucially, this does not induce collapse of the states. A particle observed after the barrier might be found to have too low energy to have crossed it in the first place. There is no consistent history of energy level, in the classical sense.

Resource tunneling is a slick way of bookkeeping fungible resources in conjunction with worldviews. Consider Alice and Bob at the candy store with \$5 split between them. Alice wants a \$3 pack of caramels from a vending machine, while Bob wants some \$2 chocolate. Our goal is to show with our bookkeeping that they have enough money to pay for their candy. It seems easy to split up the money until we find that the vending machine only takes \$5 bills. The \$5 minimum on the vending machine is like a potential barrier, and the \$3 that Alice has is her apparent energy level. Alice can resolve her plight through resource tunneling. Because the split of \$5 between Alice and Bob is not fixed, Alice can borrow Bob's \$2, exchange it for a \$5 bill, buy from the vending machine, and then give Bob back the \$2 change. This is an elegant solution, and we want our bookkeeping to reflect it without tracking every transaction between Alice and Bob. However, assigning them a \$5/\$0 split conflicts with Alice leaving with \$3 worth of candy and Bob leaving with \$2 worth. Nor does assigning them a \$3/\$2 split suffice, as Alice would temporarily have -\$2 at the vending machine. This conundrum is solved by resource tunneling.

The trick to bookkeeping these resources with tunneling is to consider \textit{both} worldviews: the \$5/\$0 split and the \$3/\$2 split. Because the first worldview shows that Alice \textit{could} be carrying \$5 for the vending machine, we know that the \$3/\$2 worldview can lend to that state, pay, and then unlend, paying only the net cost in total. It is therefore okay that the \$3/\$2 worldview temporarily assigns Alice -\$2 at the vending machine, because this lending argument shows the negative money never need be realized, and is only there for bookkeeping. All we need to check to apply this reasoning is that at every point in time there is a worldview witnessing that no one has negative money. No single worldview is always the true worldview.

\paragraph{Type System}
To use resource tunneling with our type system, we introduce \textit{remainder contexts}, which are related to \textit{IO-contexts} from linear logic proof search \cite{hodas1994logic,cervesato1996efficient} and uncomputation from quantum computing \cite{bichsel2020silq}. While typing contexts give types carrying pre-evaluation potential, remainder contexts give types carrying post-evaluation potential. Because every variable is tracked in both, this allows us to reason about the net potential paid out of each variable. We also require that every subexpression is given a variable label, which is easily achieved by re-writing expressions into let-normal form. 

To reason with remainder contexts, in type rules where some variable's value is destructed into parts, the remainder contexts should take the remainders of the parts and repack them under the original variable name. This process "uncomputes" these values to conserve total potential while removing its dependency on extraneous variables, thereby allowing tighter analysis. Then, in typing rules where some variable is used to construct a value, the needed potential of the variable should be lazily split off. Finally remainder contexts are reified in function types, so that function types track the remnants of their inputs. In contrast, for AARA, the only potential that leaves a function is on its return type, so resources returned to the input are lost. This pairing of a reconstructed input with output closely parallels Bennett's construction in quantum computing, where the same pairing is done to preserve invertibility \cite{bennett1973logical}.

\paragraph{Example}
Now let us apply resource tunneling to the function \verb"size". The call stack only grows up to the depth of the tree, so take the input to have 1 unit of depth potential. Stack also is all returned after use, so we want the remainder context to assign the same potential. When pattern matching the input tree, we get 1 unit of potential to pay for the stack frame of the recursive call, and create two subtrees. Now consider the superposition of assigning all potential to the left subtree, and all to the right. When it is all on the left, this justifies the left traversal in the usual way. When it is all on the right, we already know the recursive call on the left subtree is justified from the other worldview, so resource tunneling allows us to only consider paying the net cost of 0 stack. Then the right traversal is justified in the usual way. This allows the analysis to find the desired depth bound.

\subsection{Technical Implications}

\paragraph{Potential Functions}

Using worldviews, one can support new parameters for potential functions by varying how potential is passed along inductive datatypes. The most meaningful of these new metrics is tree depth. It was previously known that, when destructing a tree, there are a couple of different options about how to assign potential to subtrees. In linear AARA, the potential annotation is copied from the tree to its subtrees, giving potential based on the number of nodes. However, sharing potential convexly between subtrees gives potential based on tree depth, because the extreme case assigns all potential along the longest path in the tree. However, while AARA can \textit{destruct} trees in these ways, the challenge is to \textit{construct} trees with such potential. 

Worldviews justify the construction of such convexly-shared potential through convex means. Each of a set of worldviews is used to show that all potential could be allocated to one extreme. With the extreme amounts of potential justified, any convex combination is also justified. In the case of a single unit of potential per tree depth, this looks like the following: In the first worldview, the left subtree in the construction has potential equal to its depth while the right subtree has none, and vice versa in the second worldview, where all other annotations in the context are held constant. At least one of these subtrees witnesses the actual maximum tree depth, so we can be assured that we need only 1 more unit of potential to assign the constructed tree potential equal to its depth. Without multiple worldviews of potential assignments to the same trees, this kind of justification would not be possible.

\paragraph{Sharing} 

Remainder contexts remove the need for the sharing constructs of AARA. In AARA, the types are linear (technically affine), in that variables can only be used once. Were they not linear, infinite potential could be mined by repeatedly reusing a variable carrying some positive amount. However, it is standard coding practice to reuse variables. To reconcile these concerns, AARA infers locations to insert a construct that \textit{shares} the potential on one variable binding between two fresh ones that refer to the same value. This conserves potential, while providing different variable bindings at each use.

The reason sharing constructs are not needed with remainder constructs is that remainder constructs lazily split off potential as needed from variables. This is actually more flexible than the sharing construct, because the sharing construct performs its split eagerly instead. As a result, remainder contexts alone can solve the sequential map problem from \cref{subsec: aara}. The type system with remainder contexts can recognize that there is no net cost to applying $f$, allowing \verb"map f" to be typed as something like $\langle L^1(\tau); 2 \rangle \rightarrow \langle L^0(\rho);2 \rangle$, with the caveat all potential remains on the input list after application. This effectively leaves the potential in place on the input list for the subsequent application of $g$, rather than eagerly removing it at the first application.

Because the sharing construct supports no other features like e.g. resource tunneling, nor has the same depth of principled relation to logics via IO-contexts \cite{hodas1994logic,cervesato1996efficient}, remainder contexts are a more precise abstraction. 

The utility of remainder contexts stands even without worldviews. Through private correspondence with other researchers, remainder context ideas have already been used in the implementation of the Nomos typechecker \cite{das2019resource}. Other linear type systems, such as those for quantum computing, may also be able to make use of this abstraction. 

\section{Type System}
\label{sec: type}

To focus on the novelties of our type system, we present a version slimmed down to only the essential parts. We restrict the potential functions considered to only be linear, rather than any more general functions that AARA can handle. We also remove discussion of simpler types like sums, products. We even exclude lists, as they are simpler than trees. Nonetheless, our system does extend to these omissions, and the implementation in \cref{subsec: imp} makes use of the full feature set.

\subsection{Types}

The base types supported in this article are given in \cref{fig: types}, and include functions, trees, and basic types like \verb"bool" and \verb"unit". Previously, we wrote types like $\langle L^1(\tau); 1 \rangle \rightarrow \langle L^0(\rho); 0 \rangle$. However, including potential annotations on these types will quickly become cumbersome, as they will be be repeated many times across many worldviews. Thus, we abstract the many annotations into their own space: the \textit{annotation context}. This is accomplished with \textit{location indices} $I(\tau)$, which are a set of strings defined in \cref{fig: iset} representing the locations on the base type $\tau$ that can hold annotations. The annotation context then maps from these location indices to potential annotations, which keeps potential bookkeeping separate from base type dynamics. A potential-carrying type in our system can then be represented by pairing type with annotation context like $\langle T(\tbool); * \mapsto 0, \mathsf d \mapsto 1, \mathsf e* \mapsto 2 \rangle$, which represents a Boolean tree carrying 0 constant potential, 1 unit of potential per depth, and 2 units of potential per element. In general, the potential such a type indicates is defined by $\Phi$ in \cref{fig: pot}. Further in this section, we detail each type.

For convenience, we also introduce some shorthand for annotation contexts. Explicitly, an annotation context $P$ takes a location index and worldview and returns a rational. For argument clarity, location index arguments will be given in parentheses, and worldview arguments in subscript. To refer to domain, the worldviews in the domain of $P$ will be indicated by $\wv(P)$ and the location indices by $\li(P)$. To project out the submappings of $P$ taking inputs extending from location index $i$, we write $\pi_i(P)$. To instead project out the submapping of all but $\pi_i(P)$, we use $\pi_{\neg i}(P)$. Finally, we extend this projection notation to sets of indices as well.

\begin{figure}
\begin{small}
$\tau ::= \basic \mid T(\tau) \mid \tau_1 \rightarrow \tau_2$
\end{small}
\caption{Base Types}\label{fig: types}
\end{figure}

\begin{figure}
\begin{small}
\begin{mathpar}
	I(\tau \rightarrow \rho) =\{ * \}  \cup \{ \mathsf a \cdot i, \mathsf c \cdot i \mid i \in I(\tau) \}\cup\{ \mathsf b 
\cdot i \mid i \in I(\rho)\}
	
	I(\basic) =\{ * \} 
	
	I(T(\tau)) = \{ * , \mathsf d \} \cup \{ \mathsf e\cdot i \mid i \in I(\tau) \}
	
	I(\cdot) = \{ * \} (\cup \{ \circ \} \textrm{ if remainder context})
	
	I(x:\tau, \Gamma) = \{x \cdot i \mid i \in I(\tau) \} \cup I(\Gamma) 
\end{mathpar}
\caption{Location Indices}\label{fig: iset}
\end{small}
\end{figure}

\begin{figure}
\begin{small}
\begin{mathpar}	
	\Phi(v: \langle \tau; P_w \rangle) = \sum_{i \in I(\tau)} P_w(i) \cdot \phi_i(v)
	
	\phi_*(v) = 1
	
	\phi_{i\neq*}(\leaf) =  \phi_{i\neq*}(\closure f V x e) = \phi_{i\neq*}(\basic) = 0 
	
	\phi_{\mathsf d}(\node {t_1} v {t_2}) = 1 + \max(\phi_{\mathsf d}(t_1), \phi_{\mathsf d}(t_2))
	
	\phi_{\mathsf e \cdot i}(\node {t_1} v {t_2}) = \phi_i(v) + \phi_{\mathsf e \cdot i}(t_1) + \phi_{\mathsf e \cdot i}(t_2)
	
	\Phi(t: \langle T(\tau);P_w \rangle) = P_w(*) + P_w(\mathsf d) \cdot \phi_\mathit{d}(t) + \sum_{s \in t} \Phi(s: \langle \tau; \pi_e(P_w)\rangle)
	
	\Phi(V: \langle \Gamma; P_w \rangle ) = P_w(*) + \sum_{v:x \in V:\Gamma} \Phi(v : \langle \Gamma(x); \pi_x(P_w) \rangle)
	
	\Phi(v: \langle \tau, P \rangle) = max_{w\in \wv(P)} \Phi(v: \langle \tau, P_w \rangle)
\end{mathpar}
\end{small}
\caption{Potential on Typed Values}\label{fig: pot}
\end{figure}

\begin{figure*}
\begin{small}
\begin{mathpar}
\inferrule[T-Relax]{
	\Gamma \mid R \vdash e: \tau \mid  S
	\\
	\Gamma \vdash P \sub R
	\\
	 \Gamma, \circ: \tau \vdash S  \sub Q
}{
	\Gamma \mid P  \vdash e: \tau \mid Q
}

\inferrule[T-Var]{
	\pi_x(P) = \pi_x(Q) + \pi_\circ(Q)
	\\
	\pi_{\neg x} (P) =  \pi_{\neg \{x,\circ\}} (Q) 
}{
	\Gamma, x:\tau \mid P \vdash x: \tau \mid Q
}

\inferrule[T-Tick]{
	P(*) - r = Q(*)
	\\
	\pi_{\neg *} (P) =  \pi_{\neg \{*,\circ\}} (Q) 
}{
	\Gamma \mid P \vdash \tick r : \unit \mid Q
}

\inferrule[T-Let]{
	\Gamma \mid P \vdash e_1 : \tau \mid R
	\\
	\Gamma, x:\tau \mid R[x/\circ] \vdash e_2 : \rho \mid Q
	\\
	\pi_x(Q) \geq 0
}{
	\Gamma \mid P \vdash \elet x {e_1} {e_2} : \rho \mid \pi_{\neg x}(Q)
}

\inferrule[T-Superposition-In]{
	\Gamma \mid P, u \mapsto P_w  \vdash e : \tau \mid Q
}{
	\Gamma \mid P \vdash e : \tau \mid Q
}

\inferrule[T-Superposition-Out]{
	\Gamma \mid P \vdash e: \tau \mid Q
}{
	\Gamma \mid P \vdash e: \tau \mid Q, u \mapsto Q_w
}

\inferrule[T-Collapse-In]{
	\Gamma \mid P \vdash e  : \tau \mid Q
}{
	\Gamma \mid P, w \mapsto R \vdash e  : \tau \mid Q
}

\inferrule[T-Collapse-Out]{
	\Gamma \mid P \vdash e:\tau \mid Q, w \mapsto R
}{
	\Gamma \mid P \vdash e:\tau \mid Q
}

\inferrule[T-Fun-Rec]{
	\Gamma, f: \tau \rightarrow \rho, x:\tau \mid Q \vdash e : \rho \mid R
	\\
	\lfloor \pi_{\neg \{x,*\}}(Q) \rfloor
	\\
	\pi_{\neg\{x,*\}}(Q) = \pi_{\neg \{x,*,\circ\}}(R)
	\\
	\pi_x(Q)[* \mapsto Q(*)] = \pi_{f \mathsf a}(Q)
	\\
	\pi_\circ(R) = \pi_{f\mathsf b}(Q)
	\\
	\pi_x(R)[ * \mapsto R(*)] = \pi_{f\mathsf c}(Q)
	\\
	\pi_\circ(P) = \pi_f(Q)
}{
	\Gamma \mid \pi_{\neg \circ} (P) \vdash \fun f x e : \tau \rightarrow \rho \mid P
}

\inferrule[T-App]{
	\exists w \in \wv(P). \; \tau \vdash \pi_x(P_w)[* \mapsto P_w(*)] \sub \pi_{f\mathsf a}(P_w) \wedge P_w \geq 0
	\\
	\pi_{x}(P)[x*\mapsto P(*)] - \pi_{f\mathsf a}(P) + \pi_{f\mathsf c}(P) = \pi_{x}(Q)[x*\mapsto Q(*)]
	\\
	\pi_{f\mathsf b}(P) = \pi_\circ(Q)
	\\
	\pi_{\neg x}(P) = \pi_{\neg\{x,\circ\}}(Q)
}{
	\Gamma, f:\tau \rightarrow \rho, x:\tau \mid P  \vdash f\;x : \rho \mid Q
}

\inferrule[T-Leaf]{
}{
	\Gamma \mid P \vdash \leaf : T(\tau)  \mid P
}

\inferrule[T-Cond]{
	\Gamma, b:\tbool \mid P \vdash e_1 : \tau \mid Q
	\\
	\Gamma, b:\tbool \mid P \vdash e_2 : \tau \mid Q
}{
	\Gamma, b:\tbool \mid P \vdash \ite b {e_1} {e_2} : \tau \mid Q
}

\inferrule[T-Node]{
	P \RHD_{s,t_1,t_2,\circ} Q
}{
	\Gamma, s:\tau, t_1 : T(\tau), t_2 : T(\tau) \mid P \vdash \node {t_1} s {t_2} : T(\tau) \mid Q
}

\inferrule[T-Match-Tree]{
	\Gamma,t: T(\tau) \mid P \vdash e_1: \rho \mid Q
	\\
	\Gamma,t: T(\tau),  s:\tau, t_1 : T(\tau), t_2 : T(\tau) \mid R \vdash e_2: \rho \mid S
	\\
	P \lhd_{s,t_1,t_2,t} R
	\\
	S' \RHD_{s,t_1,t_2,t} Q'
	\\
	\pi_{\neg t}(S') =\pi_{\neg t}(S) 
	\\
	\pi_{\neg t}(Q') =\pi_{\neg t}(Q) 
	\\
	\pi_t(S)-\pi_t(S') = \pi_t(Q)-\pi_t(Q')
}{
	\Gamma, t: T(\tau) \mid P \vdash \matcht t {e_1} {t_1} s {t_2} {e_2} :\rho \mid Q
}

\end{mathpar}
\end{small}
\caption{Type Rules}\label{fig: rules}
\end{figure*}

The function type $\langle \tau \rightarrow \rho; P\rangle$ represents a function with argument type $\tau$ and output type $\rho$. A closure $\closure f V x e$ of this type carries only constant potential at location index $*$, and otherwise uses its type to track how potential changes across its evaluation. Location indices of the form $\mathsf a \cdot i$ track the potential of the input, $\mathsf b \cdot i$ track the potential of the output, and $\mathsf c \cdot i$ track the potential of the remainder. Since both $b*$ and $c*$ track constant potential returned from the function, we make the convention that $P(b*)=0$. We also stress one feature of function types necessary for soundness: $\pi_a(P)$ must give the same mapping for every worldview. Were this value to vary, the mechanics of resource tunneling would allow sound function types to justify unsound function types, confusing the notion of a function's cost.

The tree type $\langle T(\tau); P \rangle$ represents a tree carrying elements with base type $\tau$. While the potential carried by trees can be calculated using the definition for $\phi$ in \cref{fig: pot}, we also provide an equivalent direct shortcut to the potential in the same figure: the non-constant potential carried on a tree is $P(\mathsf d)$ times its depth plus the sum of the potential of its elements as indicated by $P(e\cdot i)$. When those elements carry constant potential, this can be used to count the nodes of the tree. While we mostly allow for negative potential amounts, we do require that $P(\mathsf d)$ is always non-negative. Were it negative, the type system could unsoundly gain potential by misattributing the depth potential to the shortest path in the tree, rather than the longest.

We extend potential to contexts in a summative fashion. The location indices of contexts are merely the names of the variables contained therein, with one exception: remainder contexts have an extra location index $\circ$ acting like a variable name for the expression preceding the context. Treating the expression itself as if it were part of the remainder context is natural in our type system, as their types share the same set of worldviews, and their potential is usually summed. Finally, because contexts are already equipped with constant potential at index $*$, we use the convention that $P(x*)=0$ for every variable $x$ in the context.

\subsection{Rules}
\label{subsec: rules}

Our type rules use judgments of the following form. The judgment means that the base types given to variables in $\Gamma$ justify that the expression $e$ is of base type $\tau$. It further means that the potential assignments to $\Gamma$'s types across the worldviews in $\wv(P)$ justify the collection of potentials assigned to the remainders and expression types across each of the worldviews in $\wv(Q)$. Note that the worldviews of $P$ need not match up to the worldviews of $Q$, as one worldview may justify or be justified by multiple worldviews.
$$\Gamma \mid P \vdash e : \tau \mid Q$$

Our type rules are then given in \cref{fig: rules}, and their interesting features explained in this section. Every rule concluding $\Gamma \mid P  \vdash e: \tau \mid Q$ includes the implicit premiss that $P_u, Q_w\geq 0$ for some worldviews $u,w$. We also use the convention that arithmetical expressions on annotation contexts, like $P+Q$, are interpreted pointwise over worldviews and location indices, where the worldviews and location indices are identical for both $P$ and $Q$. However, there is one exception: arithmetic on function types instead asserts the equality of the summands and the sum. This exception comes about because the annotations on function types do not
express how much potential the function carries, but rather how it transforms potential. The arithmetic notation allows us to easily express the redistribution of carried potential, but such expression does not have the ability to make functions cheaper or more expensive.

The \textit{T-Relax} rule is used to throw away unneeded potential. To discard potential, we use a subtyping judgment defined in \cref{fig: sub}. The judgment $\tau \vdash P \sub Q$ means that the locations indices of $P$ and $Q$ are both that of $I(\tau)$, and that the annotations of $P$ impose a harder potential requirement than $Q$.  \footnote{One might also write this subtyping judgment as $\langle \tau; P \rangle \sub \langle \tau; Q \rangle$, as the pair $\langle \tau; P \rangle$ is a type in our system. However, the base types remain constant, so this notation emphasizes that the condition really applies to the annotations $P$ and $Q$, while also saving characters. We thank an anonymous reviewer for the suggestion.} This usually means that $P$ carries more potential. While losing potential with this type rule would usually make the analysis worse, it may be productive when matching up annotations from differing code branches. 

The \textit{T-Tick} rule captures the behaviour of the tick expression used to indicate cost. The programmer may insert tick expressions throughout their code to model the resource usage of whatever resource they are interested in. Wherever $r$ resources are consumed in the program, the programmer inserts $\tick r$. When $r$ is negative, this indicates that resources are returned. The type rule matches this behaviour by removing the appropriate amount of constant potential from that ambient in the typing context. 

The superposition and collapse rules are a collection of structural rules for working with worldviews. The superposition rules say that you may copy an existing worldview to evolve differently later. The collapse rules say that you may throw away unnecessary worldviews, or those with an unsalvageable potential assignment. With these rules, it is not necessary to design any other type rule to add or remove worldviews. Nonetheless, other rules may have typing synergy with the superposition and collapse rules. For instance, when typing a branching expression like \textit{T-Cond}, where each branch is optimally typed by a disjoint set of worldviews, one can first superpose the collections together, and then when typing each branch, collapse away the unneeded worldviews. This allows for a high degree of specialization in how potential is distributed, while still initially providing each branch the same set of worldviews.

The \textit{T-Fun-Rec} rules shows how to define recursive functions. The rule does not allow for any potential to be captured in the function closure, so that functions carry 0 potential. This is ensured by applying the unary relation $\lfloor \cdot \rfloor$ to the function body's typing context, which zeroes out every non-function annotation. (Functions are the
exception here for the same reason they are excepted when performing arithmetic on types.)

The \textit{T-App} rule internalizes the rule of resource tunneling. It checks that there exists a worldview $w$ with a classically valid potential assignment ($P_w \geq 0$) in which enough potential is allocated to the argument. It then proceeds with its bookkeeping by only making each world pay the net cost. 

The tree rules have the most complex potential bookkeeping, which is abstracted away with $\lhd$ or $\RHD$ relations on annotation contexts. The formal requirements of each such relation can be found in \cref{fig: bookkeep}, and mostly concern the depth potential. Whenever a tree is destructed, like in a pattern match, the $\lhd$ relation says that depth potential is split up convexly between the subtrees. Because the worst case is that all depth potential is assigned to the deepest subtree, this procedure cannot unsoundly gain potential. On the other hand, whenever a tree is constructed, either through a $\mathsf{Node}$ expression or in the remainder context of a pattern match, the $\RHD$ relation must do more work. The relation justifies depth potential by finding worldviews witnessing the extremes of assignment: one where all the potential is on the left subtree, and one where it is all on the right. One of these subtrees witnesses the maximum depth, justifying the assignment even without knowing statically which. Note that while this process for reasoning about depth does not require the addition or removal of worldviews, neither does it require the number remains constant - it is a natural heuristic to use two distinct witnessing worldviews for each single worldview justification, which will of course halve the number of worldviews in the conclusion. The meaning of these relations is detailed more in the proof in \cref{sec:soundproof}.

\begin{figure}
\begin{mathpar}

\inferrule[S-Arr]{
	P(*) \geq Q(*)
	\\
	 \tau \vdash \pi_{\mathsf a}(Q) \sub  \pi_{\mathsf a}(P) 
	\\
	\rho \vdash \pi_{\mathsf b}(P) \sub\pi_{\mathsf b}(Q)
	\\
	\tau \vdash \pi_{\mathsf c}(P)  \sub \pi_{\mathsf c}(Q) 
}{
	\tau \rightarrow \rho \vdash P \sub Q 
}

\inferrule[S-Tree]{
	P(*) \geq Q(*)
	\\
	P(\mathsf d) \geq Q(\mathsf d)
	\\
	\tau \vdash \pi_{e} (P)  \sub \pi_{e} (Q)
}{
	T(\tau) \vdash P  \sub Q
}

\inferrule[S-Basic]{
	\tau \in \basic 
	\\
	P(*) \geq Q(*)
}{
	\tau \vdash P  \sub Q 
}

\inferrule[S-Context]{
	P(*) \geq Q(*)
	\\
	\forall x \in \mathsf{dom}(\Gamma). \;  \Gamma(x) \vdash \pi_x(P)  \sub \pi_x(Q) 
}{
	 \Gamma \vdash P  \sub Q
}

\end{mathpar}
\caption{Subtyping}\label{fig: sub}
\end{figure}

\begin{figure*}
\begin{scriptsize}
\begin{align*}
P \lhd_{s,t_1,t_2,t_3} Q \equiv \; & \pi_{\neg\{s,t_1,t_2,t_3,*\}}(P) = \pi_{\neg\{s,t_1,t_2,t_3,*\}}(Q) \wedge R=\pi_{t_3}(P) - \pi_{t_3}(Q) 
\\
&\wedge \pi_s(Q)[* \mapsto R(e*)] = \pi_{t_1\mathsf e}(Q) =\pi_{t_2\mathsf e}(Q) = \pi_e(R)
\\
& \wedge R(\mathsf d) = \pi_{t_1\mathsf d}(Q) + \pi_{t_2\mathsf d}(Q) \wedge Q(*) = P(*) + R(\mathsf e *) + R(\mathsf d)
\\
P \RHD_{s,t_1,t_2,t_3} Q \equiv\;  & \forall u \in \wv(Q). \; \exists v,w \in \wv(P).\; \pi_{\neg\{s,t_1,t_2,t_3,*\}}(P_w) = \pi_{\neg\{s,t_1,t_2,t_3,*\}}(P_v) = \pi_{\neg \{s,t_1,t_2,t_3, *\}}(Q_u)
\\
& \wedge R_v = \pi_{\{s,t_1,t_2,*\}}(P_v) - \pi_{\{s,t_1,t_2,*\}}(Q_u) \wedge R_w = \pi_{\{s,t_1,t_2,*\}}(P_w) - \pi_{\{s,t_1,t_2,*\}}(Q_u) 
\\
& \wedge \pi_{s}(R_v)[*\mapsto Q_u(t_3\mathsf e*)] = \pi_{t_1 \mathsf e}(R_v) = \pi_{t_2 \mathsf e}(R_v)=\pi_{s}(R_w)[*\mapsto Q_u(t_3 \mathsf e*)]  = \pi_{t_1 \mathsf e}(R_w) = \pi_{t_2 \mathsf e}(R_w) = \pi_{t_3\mathsf e} (Q_u)
\\
& \wedge R_v(t_1\mathsf d)=  R_w(t_2\mathsf d) =Q_u(t_3 \mathsf d) \wedge R_v(t_2\mathsf d)=  R_w(t_1\mathsf d) = 0
\\
&\wedge R_v(*) = R_w(*) = Q_u(t_3 \mathsf d) + Q_u(t_3 \mathsf e *)
\end{align*} 
\caption{Annotation Relations $\lhd$ and $\RHD$}\label{fig: bookkeep}
\end{scriptsize}
\end{figure*}

While the rules given in this section are declarative, it is not too onerous to adapt them for type inference. Base types may be found via Hindly-Milner unification \cite{damas1982principal}, and the linear constraints required may be gathered and later discharged by a linear program solver. The worldviews witnessing non-negativity - those $u,w$ such that $P_u,Q_w \geq 0$ - may also be dealt with by adding non-negative linear constraints for chosen worldviews. This leaves the effective generation and choice of such worldviews as the main obstacle to practical type inference. We discuss this in \cref{subsec: imp}.

\subsection{Examples}
\label{subsec: ex}

We now show the type system in action on a few examples. To make the derivations human-readable, we make some simplifications. Firstly, we elide easily-checkable premisses, like the many numerical conditions in the rules for trees. Secondly, we leave structural rules like relaxing, superposition, and collapse implicit - one can tell where the latter two are used by the (dis)appearance of worldviews. Thirdly, we truncate superfluous information, like the full expression being typed and full context. Finally, and most substantially, we put annotations back onto types themselves, rather than separated into an annotation context. 

The notation with these annotated types can be described as follows: At a given annotation location, the annotations corresponding to differing worldviews will be given sequentially. For trees, we put per-node potential in a superscript, and depth potential in a subscript. Thus, to represent a boolean tree with 1 unit of potential per node in one worldview and 2 units of potential per depth in another, we avoid writing $\langle T(\tbool); (1,*)\mapsto 0, (2,*)\mapsto 0, (1,\mathsf d) \mapsto 0, (2,\mathsf d) \mapsto 2, (1,\mathsf e*)\mapsto 1, (2,\mathsf e*)\mapsto 0 \rangle$, and instead write $T^{1,0}_{0,2}(\tbool)$. For functions, we put the amount leftover on the input in a new subterm after a tilde, so a function typed as $\langle \tau; a \rangle \rightarrow \langle \rho; b \rangle \sim \tau'$ takes an input with potential given by the annotated type $\tau$ and leaves it with $\tau'$ after. Remainder contexts are put following the type of the expression. The only annotation we cannot associate cleanly to a type is the ambient potential of the context, so we leave it in the position of the annotation context. 

This new notation results in judgments of the form $\Gamma \mid p \vdash e : \tau; \Delta \mid q$, where $\Gamma, \Delta$ are the typing and remainder contexts of potential-annotated types, $p,q$ are the corresponding ambient potentials of the contexts, $e$ is the expression being typed, and $\tau$ is the potential annotated type of $e$. Note that this is merely a rearrangement of the typing judgment defined in \cref{subsec: rules}.

\begin{figure}
\begin{subfigure}{0.49\textwidth}
\centering
\begin{verbatim}
let rec choc wallet = 
  match wallet with 
  | Leaf -> ()
  | Node(_,t1,t2) ->
    let () = tick{2.0} in
    let () = choc t1 in
    choc t2
\end{verbatim}
\caption{buying chocolate}
\label{fig: buy}
\centering
\begin{verbatim}
let rec caramel wallet = 
  match wallet with 
  | Leaf -> ()
  | Node(_,t1,t2) ->
    let () = tick{5.0} in
    let () = caramel t1 in
    let () = caramel t2 in
    tick{-2.0}
\end{verbatim}
\caption{vending caramel}
\label{fig: vend}
\end{subfigure}
\begin{subfigure}{0.49\textwidth}
\centering
\begin{verbatim}
let buyCandy wallet = 
  let alice = wallet in
  let bob = wallet in
  let () = caramel alice in
  choc bob
\end{verbatim}
\caption{Alice and Bob}
\label{fig: anb}
\centering
\begin{verbatim}
let badBuy wallet = 
  let alice = wallet in
  let bob = wallet in
  let () = choc bob in
  caramel alice
\end{verbatim}
\caption{worse Alice and Bob}
\label{fig: bad}
\end{subfigure}
\caption{Alice and Bob program model}
\label{fig: excode2}
\end{figure}

\paragraph{Alice and Bob} To show how the type system makes use of worldviews, remainder contexts, and resource tunneling, we model the candy store scenario from \cref{sec: rt} with the code in \cref{fig: excode2}. This code models a money-holder as a tree, and the amount of money held as the amount of potential stored on each node of the tree. Our type system can be used to show that \$5 is enough money for them to buy their candy, despite the wrinkles induced by the vending machine. It will do so by typing \verb"buyCandy" to take an input of $T^5_0(\tau)$ with no additional potential.

AARA by itself can type \verb"choc" as $\langle T^2_0(\tau); 0 \rangle \rightarrow \langle \unit ;0 \rangle$ and \verb"caramel" as $\langle T^5_0(\tau); 0 \rangle \rightarrow \langle \unit ;0 \rangle$. Remainder contexts can then be used to show that \verb"choc" leaves its input with no potential remaining, and \verb"caramel" leaves its input with 3 units of potential per node. With these typings in $\Gamma$, \verb"buyCandy" can be typed with the following two-part derivation.

Part 1:
\begin{tiny}
\begin{mathpar}
\infer{
	\Gamma \mid 0 \vdash \fun {buyCandy} {wallet} {\mathtt {let} ...} : \langle T^5_0(\tau); 0 \rangle \rightarrow \langle \unit; 0 \rangle \sim  T^0_0(\tau); \Gamma \mid 0
}{ 
	\infer {
		\Gamma, wallet:T^{5,5}_{0,0}(\tau) \mid 0,0 \vdash \elet {alice} {wallet} {...} : \unit ; \Gamma, wallet:T^{0}_{0} \mid 0
	}{
		\infer {
			\Gamma, wallet:T^{2,0}_{0,0}(\tau), alice:T^{3,5}_{0,0}(\tau) \mid 0,0  \vdash \elet {bob} {wallet} {...} : \unit ; \Gamma, wallet:T^{0}_{0}, alice:T^{0}_{0}(\tau) \mid 0
		}{
			\infer {
				\Gamma, wallet:T^{0,0}_{0,0}(\tau), alice:T^{3,5}_{0,0}(\tau), bob: T^{2,0}_{0,0}(\tau) \mid 0  \vdash \elet \_  {\app {caramel} {alice}} {...}  : \unit ; \Gamma, wallet:T^{0}_{0}, alice:T^{0}_{0}(\tau),  bob: T^{0}_{0}(\tau) \mid 0
			}{
				\infer{
					\Gamma, wallet:T^{0,0}_{0,0}(\tau), alice:T^{3,5}_{0,0}(\tau), bob: T^{2,0}_{0,0}(\tau) \mid 0,0  \vdash  \app {caramel} {alice}: \unit ; \Gamma, wallet:T^{0,0}_{0,0}, alice:T^{0,2}_{0,0}(\tau),  bob: T^{2,0}_{0,0}(\tau) \mid 0,0
				}{
				}
				\\
				part \; 2
			}
		}
	}
}

\end{mathpar}
\end{tiny}

Part 2:
\begin{tiny}
\begin{mathpar}
\infer{
	\Gamma, wallet:T^{0}_{0}, alice:T^{0}_{0}(\tau),  bob: T^{2}_{0}(\tau) \mid 0 \vdash \app {choc} {bob}: \unit; \Gamma, wallet:T^{0}_{0}, alice:T^{0}_{0}(\tau),  bob: T^{0}_{0}  \mid 0
}{
}
\end{mathpar}
\end{tiny}

Part 1 of the derivation terminates with the use of resource tunneling at the application of \verb"caramel". As assumed, $\Gamma$ types \verb"caramel" as $\langle T^5_0(\tau); 0 \rangle \rightarrow \langle \unit ;0 \rangle \sim T^3_0(\tau)$ in every worldview, while \verb"alice" is left with 5 units of potential per node in the second worldview. Because \verb"alice"'s type satisfies the input of \verb"caramel" in the second worldview, \verb"alice" is justified in every worldview to pay only the net cost of 3 units of potential per node. 

Part 2 then begins after having collapsed away the second worldview, as it is no longer needed. Note how this causes the number of worldviews in a judgment to differ between its two contexts in some lines of part 1. The rest of the derivation then proceeds with no surprises.

Note that \verb"badBuy" cannot be given the same typing due to order-dependent effects on the amount of available money. Were \verb"wallet" to start with 5 units of potential per node, then only 3 units of potential per node would remain between \verb"alice" and \verb"bob" by the time \verb"caramel" is called, which is not enough. This is the correct behaviour - were Bob to buy his chocolate first, they would not have enough money to put into the vending machine. 

\paragraph{Binary Search}

To show how we can use worldviews to work with depth-based potential, consider \cref{fig: bin2}, a slight modification of the code for binary search in \cref{fig: bin}. In our modification, we simplify the argument $x$ to be baked in as 0, and we associate cost with the number of pattern matches by inserting a tick expression. Running this function should then have a worst case cost of its argument's depth plus one. This is reflected in the following multi-part derivation typing \verb"bin0" as $\alpha = \langle T^0_1(\mathbb{Z}); 1 \rangle \rightarrow \langle \tbool; 0 \rangle \sim T^0_0(\mathbb{Z})$.

\begin{figure}
\centering
\begin{verbatim}
let rec bin0 t = let _ = tick{1.0} in
  match t with 
  | Leaf -> false
  | Node(v,t1,t2) -> if 0 = v then true
    else if 0 < v then bin0 t1
    else bin0 t2
\end{verbatim}
\caption{modified binary search}
\label{fig: bin2}
\end{figure} 

Part 1:
\begin{tiny}
\begin{mathpar}
\infer{
	\cdot \mid 1 \vdash \fun {bin} {t} {\mathtt {let} ...} : \alpha; \cdot \mid 0
}{
	\infer{
		bin0: \alpha, t:T^0_1(\mathbb{Z}) \mid 1 \vdash \elet \_ {\tick {1.0}} {...} : \tbool; bin0: \alpha, t:T^0_0(\mathbb{Z}); \cdot \mid 0
	}{
		\infer{
			bin0:\alpha, t:T^0_1(\mathbb{Z}) \mid 1 \vdash \tick {1.0} : \unit; bin0:\alpha, t:T^0_1(\mathbb{Z}) \mid 0
		}{
		}
		\\
		part \;2
	}
}
\end{mathpar}
\end{tiny}

Part 2:
\begin{tiny}
\begin{mathpar}
\infer{
	bin0: \alpha, t:T^{0,0}_{1,1}(\mathbb{Z}) \mid 0,0 \vdash \matcht t {...} v {t1} {t2} {...}: \tbool; bin0:\alpha, t:T^{0}_{0}(\mathbb{Z}) \mid 0
}{
	\infer{
				bin0:\alpha, t:T^{0,0}_{0,0}(\mathbb{Z}) \mid 0,0 \vdash \mathtt{false} : \tbool; bin0:\alpha, t:T^{0}_{0}(\mathbb{Z}) \mid 0 
	}{
	}
	\\
	part \; 3
}
\end{mathpar}
\end{tiny}

Part 3:
\begin{tiny}
\begin{mathpar}
\infer{
	bin0: \alpha, t:T^{0,0}_{0,0}(\mathbb{Z}), v:\mathbb{Z}, t1:T^{0,0}_{1,0}(\mathbb{Z}),  t2:T^{0,0}_{0,1}(\mathbb{Z}) \mid 1,1 \vdash \ite {0=v} {\mathtt {true}} {...} : \tbool; bin0:\alpha, t:T^{0}_{0}(\mathbb{Z}), v:\mathbb{Z}, t1:T^{0}_{0}(\mathbb{Z}),  t2:T^{0}_{0}(\mathbb{Z}) \mid 0
}{
	\infer{
		bin0: \alpha, t:T^{0,0}_{0,0}(\mathbb{Z}), v:\mathbb{Z}, t1:T^{0,0}_{1,0}(\mathbb{Z}),  t2:T^{0,0}_{0,1}(\mathbb{Z}) \mid 1,1 \vdash \ite {0<v} {...} {...} : \tbool; bin0: \alpha, t:T^{0}_{0}(\mathbb{Z}), v:\mathbb{Z}, t1:T^{0}_{0}(\mathbb{Z}),  t2:T^{0}_{0}(\mathbb{Z}) \mid 0
	}{
		\infer{
			bin0: \alpha, t:T^{0}_{0}(\mathbb{Z}), v:\mathbb{Z}, t1:T^{0}_{1}(\mathbb{Z}),  t2:T^{0}_{0}(\mathbb{Z}) \mid 1 \vdash \app {bin0} {t1} :\tbool ; bin0: \alpha, t:T^{0}_{0}(\mathbb{Z}), v:\mathbb{Z}, t1:T^{0}_{0}(\mathbb{Z}),  t2:T^{0}_{0}(\mathbb{Z}) \mid 0
		}{
		}
		\\
		part\; 4
	}
}
\end{mathpar}
\end{tiny}

Part 4:
\begin{tiny}
\begin{mathpar}
\infer{
	bin0:\alpha, t:T^{0}_{0}(\mathbb{Z}), v:\mathbb{Z}, t1:T^{0}_{0}(\mathbb{Z}),  t2:T^{0}_{1}(\mathbb{Z}) \mid 1 \vdash \app {bin0} {t2}:\tbool  ; bin0:\alpha, t:T^{0}_{0}(\mathbb{Z}), v:\mathbb{Z}, t1:T^{0}_{0}(\mathbb{Z}),  t2:T^{0}_{0}(\mathbb{Z}) \mid 0
}{
}
\end{mathpar}
\end{tiny}

Part 1 of the derivation proceeds normally, paying the cost of the pattern match. Then in part 2, a second worldview is superposed prior to the pattern match. In part 3, these worldviews diverge, as the depth potential on $t$ is split in different ways between them - the first worldview puts all depth potential on $t1$, and the second puts it all on $t2$. We then see the benefit of this in the terminal steps of parts 3 and 4, where each branch collapses away a different worldview for the recursive call. The use of worldviews has allowed each branch to individually specialize the allocation of potential in the most beneficial way possible. 

Specializing potential to branches is particularly useful for working with depth, as seen above. This is because depth relates to specifically one path through a tree: the deepest path. Worldviews allow the allocation of potential to cover each path individually, including the deepest path. Without worldviews, AARA overapproximates by covering all paths simultaneously, even though only one path may be used in execution. This overapproximation associates a cost bound for the above code proportional to the total number of $t$'s nodes, which can be exponentially worse than the bound based on $t$'s depth.

\section{Soundness}
\label{sec: sound}

The type system, with all of its new quantum reasoning principles, is sound. To show this, we adopt the operational semantics used with AARA, and use it to formally assign high-water mark and net cost to code. We then prove that the potential assigned by our type system in \cref{fig: pot} satisfies the invariants of the physicist's method with respect to that cost. This means that standard physicist method reasoning applies, and our type system successfully bounds the program's cost. Because this is proven for general usage of the tick operation, it is sound not only with respect to stack costs, but any cost model implementable in AARA. We formally state our soundness theorem in \cref{subsec: proof}.

\subsection{Operational Semantics}
\label{subsec: op}

The operational semantics we use has been used with AARA in the past; see e.g. \cite{kahn2020exponential}. This operations semantics uses big-step judgments of the following form. As a precondition for our soundness theorem, we assume that we can derive such a judgment for the expression at hand.
\begin{mathpar}
	V \vdash e \Downarrow v \mid (p,q)
\end{mathpar}

This statement means that, under the mapping $V$ from variables to values, the expression $e$ evaluates to the value $v$ with cost described by $p,q$. The rational $p$ gives the high-water mark of cost, and $q$ gives the returned resources at the end, such that $p-q$ gives the net cost. See \cref{fig: cost} for the full set of rules, given for expressions in let-normal form.

\begin{figure*}[ht]
  \def \MathparLineskip {\lineskip=0.33cm}    
\begin{mathpar}
\fontsize{8}{12}
      	\infer[\mathit{Tick}]{
		V \vdash \mathit{tick}\{r\} \Downarrow () \mid (p,q) 
	}{
		p = max(r,0)
		&
		q = max(-r,0)
	}
    
    \infer[\!\!\mathit{Var}]{
        V \vdash x \Downarrow v \mid (0,0)
    }{
    	V(x) =v
    }

    \infer[\mathit{Let}]{
        V \vdash \elet x {e_1} {e_2} \Downarrow v_2 \mid (p+max(p'-q,0),q'+max(q-p',0))
    }{
        V \vdash e_1 \Downarrow v_1 \mid (p,q)
        & 
        V[x \mapsto v_1] \vdash e_2 \Downarrow v_2 \mid (p',q')
    }
    
        \infer[\mathit{CondT}]{
            V \vdash \ite {x_b} {e_t} {e_f} \Downarrow v \mid (p,q)
        }{
           	V(x_b) = \mathit{true}
		& 
		V \vdash e_t \Downarrow v \mid (p,q)
        }

        \infer[\mathit{CondF}]{
            V \vdash \ite {x_b} {e_t} {e_f} \Downarrow v \mid (p,q)
        }{
           	V(x_b) = \mathit{false}
		& 
		V \vdash e_f \Downarrow v \mid (p,q)
	}

        \infer[\mathit{App}]{
            V \vdash f\; x \Downarrow v  \mid (p,q)
        }{
        	    V(x) = v_x
	    &
            V(f) = \closure g {V'} {x'} {e}
	    &
            V'[x' \mapsto v_x, g \mapsto \closure g {V'} {x'} {e}] \vdash e \Downarrow v \mid (p,q)
        }
        
        \infer[\mathit{Fun}]{
        		V \vdash \fun f x e \Downarrow \closure f V x e \mid (0,0)
        }{
        }
        
          \infer[\mathit{Node}]{
           V \vdash \node {t_1} {a} {t_2} \Downarrow \node {v_1} {v_a} {v_2}\mid  (0,0)     
        }{
        	V(t_1) = v_1
	&
	V(t_2) = v_2
	& 
	V(a) = v_a
        }

        \infer[\!\!\mathit{TMatchL}]{
           V \vdash  \matcht x {e_1} {t_1} a {t_2} {e_2} \Downarrow v \mid  (p,q)    
        }{
            V(x) = \leaf
            &
            V \vdash e_1 \Downarrow v \mid (p,q)
        }

        \infer[\!\!\mathit{TMatchT}]{
            V \vdash  \matcht x {e_1} {t_1} a {t_2} {e_2} \Downarrow v \mid  (p,q)     
        }{
            V(x) = \node {v_1} {v_a} {v_2} 
            &
            V[t_1 \mapsto v_1, t_2 \mapsto v_2, a \mapsto v_a] \vdash e_2 \Downarrow v \mid (p,q)
        }
        
         \infer[\mathit{Leaf}]{
            V \vdash \leaf \Downarrow \leaf \mid (0,0)
        }{
        }
	\end{mathpar}
	\caption{Operational Cost Semantics Rules}
    \label{fig: cost}
\end{figure*}

We connect the values used in the operational semantics with our types using the judgment $v:\langle \tau; P\rangle$. The judgment means that the value $v$ can be given the type $\langle \tau; P \rangle$. For non-closure values $v$, this judgment is independent of $P$, and simply requires the base type to match. However, for closures to be typed as functions, it must further be the case that the body of the closure can be appropriately typed with our potential-carrying types. This includes maintaining the condition on function types that the input potential annotations are invariant across all worldviews of $P$. We extend this judgment across contexts pointwise.

\subsection{Proof Strategy} 
\label{subsec: proof}

The soundness of our type system hinges on the fact that its potential satisfies the invariants of the physicist's method. That is, taking potential as the max across the potential in each worldview, as defined in \cref{fig: pot}, we would like the initial potential to bound the high-water mark cost, and the difference in potential to bound the net cost. More formally:

\begin{theorem}[Soundness] If $V\vdash e \Downarrow v \mid (p,q)$ for $V:\langle \Gamma; P \rangle$ - so that the high-water cost is $p$ and $p-q$ is the net - then whenever our type system derives $\Gamma \mid P \vdash e : \tau \mid Q$, we find:
\begin{align*}
	p \leq & \Phi(V : \langle \Gamma; P \rangle)
	\\
	 p-q \leq & \Phi(V : \langle \Gamma; P\rangle) - \Phi(V,v : \langle \Gamma,\circ:\tau; Q \rangle )
\end{align*}
\end{theorem}

In order to prove this statement, we actually prove the given conditions entail a slightly different consequent. We then find that our desired statement follows as a corollary. The consequent we prove is: Firstly, there exists a worldview wherein the typing context potential is sufficient to cover the high-water mark cost. Secondly, for all worldviews across the remainder context, there exists a worldview in the typing context such that the change in potential between them is sufficient to cover the net cost. Formally, this is the following statement:
\begin{align*}
\exists w \in \wv(P). \; & p \leq \Phi(V : \langle \Gamma; P_w \rangle)
\\
\forall u \in \wv(Q). \; \exists w \in \wv(P). \; & p-q \leq \Phi(V : \langle \Gamma; P_w \rangle)
\\ & - \Phi(V,v : \langle \Gamma,\circ:\tau; Q_u \rangle )
\end{align*}

Unlike the original soundness statement, this new statement's extra structure allows one to prove it by nested induction over typing and cost derivations. Such induction is relatively routine, with the most interesting cases being the justification of new kinds of potential, like depth. In such cases, the rules pick out multiple worldviews that are meant to bear witness to the potential at the extremes of each allocation. For instance, for tree depth, one worldview assigns all the potential to the left subtree, and the other to the right. One of these must achieve the maximum depth, and this will yield the desired potential bound. 

With this new statement proven, our desired soundness statement follows. The maximum potential is at least that of the existential witness for high-water mark, so the maximum potential will also be sufficient to cover the high water mark. And the difference between max potentials will similarly bound the net cost. This recovers the use of the physicist's method, but with potential that behaves non-locally. The full proof can be found in \cref{sec:soundproof}.

\section{Related Work}
\label{sec: related}

No work has previously explored quantum refinements of the physicist's method, nor made use of techniques like worldviews or resource tunneling in the context of resource analysis. Advantages of our techniques include support for amortization and non-monotone resources, natural compositionality of a type system, automatability, the formal soundness proof w.r.t. a cost semantics, and type derivations that are certificates of bounds.

Most closely related is the previous work on AARA, which has been
introduced by Hofmann and Jost~\cite{hofmann2003static} to
automatically derive linear bounds on the heap-space consumption of
first-order functional programs.
It has successively been generalized to support other
resources~\cite{Jost09}, higher-order functions~\cite{Jost10},
probabilistic programs~\cite{wang2020raising}, and more complex bounds
such as polynomials~\cite{HoffmannAH12,hoffmann2017towards} and
exponential functions~\cite{kahn2020exponential}.
AARA can also be applied to imperative
languages~\cite{CarbonneauxHZ15,NgoCH17}.
In this work, we specifically focus on non-monotone resources,
introduce worldviews and resource tunneling, and show how this
improves on standard AARA in \cref{sec: ine}.

One previous work on AARA uses bunched logic to reason about tree depth in AARA-style \cite{campbell2009amortised}.
This approach generates a new kind of constraint based on axiomatic
equivalence relations on different bunches of types in the typing
context. Then, similarly to our work, their constraint system
approximates maxima using convex combinations. 
One limitation of the bunched approach is that the programmer must
indicate to the program whether tree arguments use depth- or
node-based potential, whereas our approach can infer any combination
of both simultaneously.

Another relevant work uses "give-back" annotations on types in order to recover potential for later use after temporarily allocating it elsewhere \cite{campbell2008type}. This achieves a similar goal as our remainder contexts, which return remaining potential on certain variable bindings back to the potential's source after the bindings fall out of scope. The "give-back" annotations accomplish this by pre-determining how much potential to return, and then using third-party safety analyses to determine when the temporary allocation's value is no longer in use. The "give-back" type rules cannot reason about reusing potential by themselves, and also require that the resources must be completely returnable, like memory. In contrast, remainder contexts support more general resources, and do so uniformly as a part of the type system.

Potential-based methods are also used in other techniques. A recent
work~\cite{cutler2020} uses the physicist's method to combine
amortized analysis and recurrence solving.
Other works have proposed powerful program
logics~\cite{gueneau2018fistful,timecredits,atkey2010amortised} with
credits that are similar to potential annotations.
In contrast to our focus on automatic analysis, these logic are mainly
designed for manual resource analyses.

Other type-based techniques for resource analysis include linear dependent types \cite{dal2013geometry,dal2011linear}, refinement types \cite{wang2017timl,POPL:RBG18}, and other annotated type systems \cite{POPL:CW00,POPL:Danielsson08}.
Some of these approaches, like AARA, implement a linear type system
with specialized constraints, and compose analyses via their
types. 
The richer the type system, the more
likely it will require programmer intervention in order to solve its
constraints. Some of type techniques are also specialized to handle
non-monotonic costs like space, like the sized types of
\cite{vasconcelos2008space}.

Outside of the functional setting the capabilities of techniques are harder to compare. Imperative and object-oriented methods are based in numerical analyses on a program’s integers, and do not handle data structures. Most also ignore non-monotonic resources. The imperative work on SPEED \cite{gulwani2009speed,speedloop} comes close by reducing data structures to numerical analysis, but they require manual definition and verification of the numerical reduction, and do not handle non-monotone resources. The imperative work in \cite{albert2015non} does handle non-monotonic resources automatically, but only handles integers and operates with abstract execution. The object-oriented work on COSTA \cite{albert2007cost} also comes close, but can have difficulty reasoning about conditional guards and sums the cost over branches, whereas our work doesn’t reason about guards at all and optimizes to the max branch. The object-oriented work of \cite{hofmann2006type, hofmann2013automatic} uses "views" with the physicist’s method, but its views deal with aliases and serve no role like our worldviews. Moreover, it is difficult to automate.

Other techniques for resource bound analysis range from recurrence solving
\cite{POPL:KML20,danner2015denotational,kincaid2017compositional}, to
term re-writing
\cite{RTA:AM13,hirokawa2008automated,avanzini2015analysing,naaf2017complexity,noschinski2013analyzing},
to loop analysis \cite{speedloop,kincaid2017non,LPAR:BHH10}, and more
\cite{chatterjee2019non,lopez2018interval}. No such approach is similar to ours.
%
%

\section{Implementation \& Evaluation} 
\label{sec: ine}

To test the efficacy of the developed type system, we implemented it in a prototype analyzer and compare its performance to the pre-existing implementation of AARA in RaML. We test both how accurately and how fast the two can analyze naive stack bounds for the OCaml standard library Set module. Our type system's new capabilities greatly increases the accuracy of its analysis, with only moderate performance tradeoff.

\begin{table*}
\begin{tiny}
 \begin{tabular}{c c | c c c c | c c c c} 
 && RaML & & & & Prototype & & &
 \\
 \hline
 Function & LoC & Time(s) & Constrs & Stack Bound & Returned & Time(s) & Constrs & Stack Bound & Returned
 \\ 
 \hline
 ordcompare & 3 & 0.01 & 3& 0 & 0 & 0.00 & 23 & 0 &0
 \\
 height &5 & 0.01 & 8 & 0 & 0& 0.00 & 214 &0 & 0
 \\
 create & 4& 0.02 &31 & 0 & 0 & 0.04 & 10722 & 0 & 0
 \\
 bal & 44& 0.17 & 505& 1 & 1 & 3.67 & 513451& 1 & 1
 \\
 add &62 & 1.11 & 1085 &  $n+1$ & 1 & 7.33 & 697545 & $d+1$ & $d+1$
 \\
 singleton & 1& 0.00 & 2 & 0 & 0& 0.00&548& 0 & 0
 \\
 add$\_$min$\_$elt & 52 & 0.45 & 532 & $n+1$ & 1 & 3.69 &538691& $d+1$ & $d+1$
 \\
 add$\_$max$\_$elt & 52 & 0.47 & 534 & $n+1$ & 1 & 3.85 &538763& $d+1$ & $d+1$
 \\
 join & 71& 4.29 & 2197 & $n_0+n_1+2$ & 1 & 6.43 &764565& $d_0 + d_1 + 2$ & $d_0 + d_1 + 2$
 \\
 min$\_$elt &9& 0.01 & 26 &$ n$ & 1 & 0.01 &3616& $d$ & $d$
 \\
 min$\_$elt$\_$opt &9& 0.01 & 32 & $n+1$ & 1 & 0.01 &4224&$d$ &$d$
 \\
 max$\_$elt &9& 0.01 &26 & $n$ & 1 & 0.01 &3616& $d$ & $d$
 \\
 max$\_$elt$\_$opt &9& 0.01 & 32 & $n+1$ & 1 & 0.01&4224 &$d$ &$d$
 \\
 remove$\_$min$\_$elt &54& 0.46 & 540 & $n$ & 1 & 3.68 &539700& $d$ & $d$
 \\
 merge &75& 0.81 & 1124 & $2n_1+1$ & 1 & 6.36 &601655& $d_1+1$ & $d_1 + 1$
 \\
 concat &102& 5.96 & 2816 & $n_0+2n_1+1$ & 1 & 10.57 &903682& $n_0  + n_1 + .5d_1 + 4.5$ & $n_0'+4$
 \\
 split &91& 17.46 & 4447 & fail & fail & 11.72 &941207& $n_0+4$ & $n_0'+n_1'+4$
 \\
 is$\_$empty & 1 & 0.01 & 8 & 0 & 0 & 0.01 &160& 0 & 0
 \\
 mem & 10 & 0.02 & 34 & $n$ & 0 & 0.03 &12566& $d$ &$ d$
 \\
 remove &96 & 5.03 & 2211& $2n+1$ & 1 & 8.56 &880667& $d+1$ & $d+1$
 \\
 union &127& 164.91 & 15568 & fail & fail & 25.14 &1586889& fail & fail
 \\
 inter &137& 94.03 & 9532 & fail & fail & 24.33 &1486736& $n_0+n_1+5$ & $n'+5$
 \\
 diff & 137 & 95.01 & 9535 & fail & fail & 21.40 & 1433676 &  $n_0+n_1+5$ & $n'+5$
 \\
 cons$\_$enum & 7 & 0.01 & 27 & $n_0+1$ & 1 & 0.02 & 6338 & $d_0$ & $d_0$
 \\
 compare$\_$aux & 27 & 0.13 & 509 & $n_0s_0{+}n_1s_1{+}s_1{+}1$ & 1 & 0.21 &61545& $.5s_0 {+} 1.5n_0 {+} .5s_1{+}1.5n_1{+}1$ & 1
 \\
 compare & 30 & 0.17 & 739 & $n_0+2n_1+2$ & 1 & 0.29 & 82131 & $1.5n_0+1.5n_1+2$ &	2
 \\
 equal & 33 & 0.19 & 745 &  $n_0+2n_1+3$ & 0 & 0.30 &87239&$1.5n_0+1.5n_1+3$ & 3
 \\
 subset &20 & 0.42 & 1607 & $n_0n_1$ & 0 &0.61 &156721& $d_0+d_1$ & $d_0+d_1$
 \\
 iter & 4 & 0.01 & 28 & $n+1$ & 1 & 0.04 &11161& $d$ & $d$
 \\
 fold & 4 & 0.02 & 28& $n+1$ & 1 & 0.04 &15726&$d$ & $d$
 \\
 for$\_$all & 3 & 0.02 & 32 & $n$ & 0 & 0.04 &11581&$d$ & $d$
 \\
 exists & 3 & 0.02 & 32& $n$ & 0 &0.04 &11581&$d$ & $d$
 \\
 filter & 115 & 30.88 & 5074 & fail & fail & 61.35 &2170433&$.5n^2+.5n+5$ & $n'+5$
 \\
 partition & 114 & 74.91 & 10076 & fail & fail & 14.33&1158301& $2n+5$ & $n_0'+n_1'+5$ 
 \\ 
 cardinal & 3 & 0.01 & 27& $n$ & 0 &0.01 &3763&$d$ & $d$
 \\
 elements$\_$aux & 4 & 0.01 & 32& $n_1$ & 1 & 0.02 &6953&$d$ & $d$
 \\
 elements & 7 & 0.01 & 36& $n+2$ & 1 & 0.03 &8182 &$d+1$ & $d+1$
\end{tabular}
\end{tiny}
\caption{Experimental Statistics and Inferred Stack Bounds}
\label{tab: stats}
\end{table*}

\subsection{Implementation} 
\label{subsec: imp}

Our typechecker implements the type rules presented in \cref{sec: type}, and more. The implementation extends the type rules to include support for some new code features, including the types and expressions of products, sums, and lists. It also extends tree potential from strictly linear functions of tree depth and nodes, to polynomial potential thereof. Though polynomial potential is mostly unneeded for our tests, it does get used in analyzing $\mathsf{filter}$, and its inclusion shows that our typesystem extends unproblematically to more advanced AARA features.

The typechecker works over code in a couple passes to perform its analysis. The first pass converts input code to let-normal form, so that type rules with fewer moving parts can be applied - our presentation of the rules in this work makes the same simplification. The second pass converts the code's abstract syntax tree into a type derivation tree, which is used to generate constraints for standard Hindley-Milner unification \cite{damas1982principal}. The result of unification is then converted into a derivation tree for the base type of the expression. Finally, this derivation tree is decorated with potential annotations and used to generate linear constraints. Using an off-the-shelf LP solver on the linear constraints then gives the resource bounds of the code, finishing the analysis.

To perform the necessary accounting for worldviews, the typechecker calculates how many worldviews it might need at each point in the program. After calculating these numbers, the typechecker superposes that many worldviews from the start. Once these worldviews are in place, it then becomes unnecessary to manipulate worldviews via superposition rules anywhere else in the derivation. Use of the collapse rule can also be confined to predetermined locations, such as branches. So, to calculate an upper bound on the number of worldviews needed, the typechecker makes use of the following observations. This calculation is rather pessimistic to capture the worst case for the analysis, and more aggressive heuristics could improve performance.

Firstly, every subexpression needs a worldview in which every annotation is non-negative, so as to prove that the total potential is non-negative. The typechecker can dictate which world should fulfill this role by generating the appropriate non-negative linear constraints, and the LP-solver will fill in the details. Conservatively, the typechecker could then operate by allocating a new such non-negative worldview for every subsexpression. However, one can observe that there are only two kinds of expressions that can turn a non-negative annotation into negative one: function applications (which subtract the net cost from their argument) and positive ticks (which subtract the tick amount from the constant potential). Thus, the typechecker reuses non-negative worldviews for every other kind of expression.

Secondly, some rules use two worldviews of their premisses to justify one in their conclusion, specifically those rules concerning the depth potential of trees. At worst, this means that the number of worldviews might need to double at this subexpression to capture the most general analysis possible.

Thirdly, expressions that branch might use different worldviews in each branch. The typechecker again assumes the worst here, and bounds how many worldviews are needed in each branch, and sums the bounds together. (Note that when typing each branch, the collapse rules can be used to remove the extra worldviews.) At worst, this also might double the number of worldviews needed to capture the most general analysis possible.

Finally, function types are slightly simplified by requiring that they have the same annotations in every worldview. This allows helper functions to be typed independently, at the tradeoff of not allowing functions to have different allocations of potential among their return types. In practice, this capability is not usually important to the analysis, but the tradeoff enables a large performance benefit: Namely, the exponential growth of the worldview count is contained to a single function body. So long as code is broken into helper functions, as is often the case, the exponential growth mitigated. This is supported by our performance evaluation on real-world OCaml code in \cref{subsec: eval}, which achieves reasonable times.

To summarize this worldview counting, the number of worldviews needed at a point in a program can be no worse than the number of distinct \textit{non-recursive} execution paths extending from that point, which is exponential in the expression's depth. By \textit{non-recursive}, we refer to ignoring recursive calls along the paths, which is justified because because recursive calls do not need to be retyped. While in practice one can usually derive types using many fewer worldviews than this exponential count, our implementation does not try to make this optimization. One might also see this count as arising for similar reasons to the exponential state space maintained by qubits in quantum computing.

One final optimization our implementation makes is to only apply the relax rule when branches occur. This is because the only good reason to throw away potential in the analysis is to allow the return type and remainder contexts of each branch to match in cases where one branch was more resource-efficient than the other. By including this optimization, significantly fewer non-trivial constraints generated for the LP, and the expressivity of the analysis is unaffected. 

\subsection{Evaluation}
\label{subsec: eval}

\paragraph{Experimental Setup}

To evaluate our bound inference, we compare it against RaML's bound inference, which is based on standard AARA. As test code, we use the Set module from OCaml's standard library \cite{setmodule}, which implements sets using binary trees. We compare the inferred naive stack bound for each function (if any), and measure the two implementations' performances. We find that our under-optimized prototype is able to infer significantly more precise bounds with only moderate overhead compared to RaML.

We use a naive stack metric that simply counts the number of needed
stack frames without accounting for tail-calls or other optimization.
When analyzing a function, each typechecker must also analyze all helper functions, which sometimes exceeds one hundred of lines of code. We time the implementations after they give the code base types, but include the time it takes to solve the linear programs they generate. To even the playing field, both typecheckers use COIN-OR's open source LP solver \cite{coinor} to solve their linear constraints. For parsing reasons, some of OCaml's syntactic sugar was manually desugared.

When gathering the experimental data, each implementation is run at the lowest degree setting that successfully analyzes the code. In the event of failure, the data from the linear bound search is used. This means that each usually only searches for linear bounds. Nonetheless, our implementation uses quadratic bounds for $\mathsf{filter}$, and RaML uses them for 4 functions ($\mathsf{compare\_aux}$ through $\mathsf{subset}$). These cases also show off a capability of RaML not implemented in our prototype: \textit{multivariate potential}, wherein size parameters may be multiplied together.

\paragraph{Findings}
The results of the analysis are in \cref{tab: stats}. For each of 36 functions from the Set module and our implementation of the Ord module comparison $\mathsf{ordcompare}$, we provide the following data: the number of lines of code in the test file (LoC), the time each implementation takes to infer potential annotations (Time), the number of linear constraints generated during inference (Constraints), the inferred upper bound on stack (Stack Bound), and the number of stack frames returned according to the analyses (Returned). In resource bounds, $n$ is used to describe the node count of tree arguments, $d$ the depth of tree arguments, and $s$ the size of list arguments. Subscripts disambiguate arguments by index, and a prime (') is added to refer returns instead of arguments. 

Measuring naive stack bounds on the Set module showcases multiple strengths of our type system. For one, the Set module is implemented using trees, making tree depth a relevant parameter. For another, stack frames are a resource which is returned after use. Neither of these cases are handled well in pre-existing AARA, but here they both naturally coincide. 

The data does indeed show that our analyses yields significantly tighter results. In 30 cases, our implementation finds a tighter bound than RaML, in the remaining 7 cases they find the same tight bound. Further, in 31 out of 37 cases, our implementation provides a tight bound on returned resources, whereas RaML never infers that more than 1 resource unit is returned. (Because we are measuring stack, we know that all resources should be returned in each case.) 

The general trend of \cref{tab: stats} is that our prototype can perform more accurate resource analyses at the expense of slower performance. However, on difficult-to-analyze code like the $\mathsf{join}$ function, the speed of our analysis approaches RaML's order of magnitude. This shows that our prototype can already achieve plausible performance where analysis tools are most desired. Furthermore, our implementation can analyze difficult code like the $\mathsf{split}$ function, where RaML fails to derive a bound. Our implementation only is unable to handle $\mathsf {union}$, while RaML fails on 6 different functions. 

The table also shows that the time taken by our prototype to generate linear constraints and solve them is roughly equal, so time could be improved by either a stronger LP solver or more aggressive heuristics and strategies for constraint generation. In particular, it may not be necessary to double the number of worldviews as often as we do.

That our prototype performs as quickly as it does compared to RaML is somewhat surprising. Our implementation internally maintains a completely decorated type derivation tree, while RaML does not. To keep up performance, RaML aggressively reuses annotations where possible, in general aiming to generate as few linear constraints as possible. However, as \cref{tab: stats} shows, the multiple-orders-of-magnitude more constraints generated in this fashion did not result in commensurate slowdown. We believe this is because easy constraints comprise a majority of our implementation's constraints, and modern LP solvers can eliminate them quickly. For example, constraints of the form $x=y$ comprise $55.56\%$ of generated constraints. One can find a more detailed breakdowns of the performance statistics in \cref{fig: breakdown}. Timing is subdivided into constraining and running the LP solver, and the constraint count breaks down into variable identities ($x=y$ for variables $x,y$), constant offsets ($x=y+k$ for variables $x,y$ and constant $k$), other equality constraints, and any remaining inequalities. 

\begin{table*}
\begin{tabular}{c | c c | c c c c}
	Function & Constrain Time & LP Time & Var IDs & Offsets &  Eqs & Other Ineqs
	\\
	\hline
	ordcompare & 0.00 & 0.00 & 7 & 0 & 3 & 13
	\\
	height & 0.00 &0.00 & 83	&0&	26&	105
	\\
	create& 0.00&	0.04 &5138&	0	&1205&	4379
	\\
	bal & 1.83 &	1.84 & 278669 &	262	&49868&	184652
	\\
	add & 3.32	&4.01 & 391214&	335&	60411&	245585
	\\
	singleton & 0.00 & 0.00 & 208	&0&	43&	297
	\\
	add$\_$min$\_$element & 1.94	&1.75 &292062&	279&	50773	&195577
	\\
	add$\_$max$\_$element & 1.96	&1.89&292124&	279&	50773&	195587
	\\
	join & 3.96 &	2.47 &417754&	357&	64253&	282201
	\\
	min$\_$elt & 0.00 & 0.01 &1738	&9&	392&	1477
	\\
	min$\_$elt$\_$opt & 0.00 & 0.01 &2117&	9&	403&	1695
	\\
	max$\_$elt & 0.00 & 0.01 &1738	&9&	392	&1477
	\\
	max$\_$elt$\_$opt  & 0.00 & 0.01 &2117&	9&	403&	1695
	\\
	remove$\_$min$\_$elt & 2	.00 & 1.68 & 292644&	284&	51707&	195065
	\\
	merge & 2.39 &	3.97 &325913&	332&	56022&	219388
	\\
	concat & 5.24	&5.33 &492858	&427	&73863&	336534
	\\
	split & 5.81&5.91 &523681&	414&	72934&	344178
	\\
	is$\_$empty & 0.00&	0.01&55&	0&	26&	79
	\\
	mem & 0.00&	0.03&7388&	22&	578&	4578
	\\
	remove & 5.09 &	3.46 &496624&	414&	71456&	312173
	\\
	union & 16.24 &	8.90 &894505	&625	&111227&	580532
	\\
	inter & 14.28 &	10.05&840854&	603&	100644&	544635
\end{tabular}
\caption{Prototype Statistic Breakdown}
\label{fig: breakdown}
\end{table*}

Interestingly, our prototype finds good bounds even though the Set module performs many tree operations using \textit{semantic} properties of the code, rather than structural. For instance, in $\mathsf{bal}$, trees are balanced by tracking their tree height. Our type system completely ignores that those integers track tree height, and is still able to infer good resource bounds.

\section{Conclusion}
\label{sec: conc}

In this work, we have presented the quantum physicist's method, a novel
set of quantum-physics-inspired reasoning principles to better adapt
the physicist's method of amortized analysis to reason about
non-monotone resources like stack use.
These principles, superposition and resource tunneling, have been
integrated with AARA and allow the automatic derivation of more
precise resource bounds, including new bounds based on tree
depth. Through implementation and testing, we found that our
extension to AARA requires only moderate overhead.

One challenge for future work is the adaptation of the convex
techniques for depth potential to handle other kinds of resource
bounds, like those based on the largest element in the data
structure. The difficulty is that if one has more than one notion of
element size, like depth \textit{and} node count, then the number of extreme
points to check grows exponentially. 

\begin{acks}                            
  We thank John Grosen for comments and suggestions that helped to
  improve the soundness proof and the software artifact.

  This article is based on research supported by DARPA under AA
  Contract FA8750-18-C-0092 and by the National Science Foundation
  under SaTC Award 1801369, CAREER Award 1845514, and SHF Awards
  1812876 and 2007784. Any opinions, findings, and conclusions
  contained in this document are those of the authors and do not
  necessarily reflect the views of the sponsoring organizations.
\end{acks}

\bibliographystyle{ACM-Reference-Format}
\bibliography{sources.bib}

\clearpage

\appendix
\section{Soundness Proof} \label{sec:soundproof}

Here we prove our soundness theorem:

\begin{theorem}[Soundness] If $V\vdash e \Downarrow v \mid (p,q)$ for $V:\langle \Gamma; P \rangle$ - so that the high-water cost is $p$ and $p-q$ is the net cost - then whenever our type system derives $\Gamma \mid P \vdash e : \tau \mid Q$, we find:
\begin{align*}
	p \leq & \Phi(V : \langle \Gamma; P \rangle)
	\\
	 p-q \leq & \Phi(V : \langle \Gamma; P\rangle) - \Phi(V,v : \langle \Gamma,\circ:\tau; Q \rangle )
\end{align*}
\end{theorem}

To do so, we note that it follows as a corollary of the following lemma:

\begin{lemma}[Worldview Soundness] If $V\vdash e \Downarrow v \mid (p,q)$ for $V:\langle \Gamma; P \rangle$ - so that the high-water cost is $p$ and $p-q$ is the net cost - then whenever our type system derives $\Gamma \mid P \vdash e : \tau \mid Q$, we find:
\begin{align}
\exists w \in \wv(P). \; & p \leq \Phi(V : \langle \Gamma; P_w \rangle) \label{ex}
\\
\forall u \in \wv(Q). \; \exists w \in \wv(P). \; & p-q \leq \Phi(V : \langle \Gamma; P_w \rangle) - \Phi(V,v : \langle \Gamma,\circ:\tau; Q_u \rangle ) \label{un}
\end{align} 
\end{lemma}

We prove this lemma via nested induction over the type and cost derivations. The lexicographic ordering of the cost derivation size followed by the type derivation size will always decrease. Recall throughout that, for each judgment $\Gamma \mid P \vdash e : \tau \mid Q$, we impose the side condition that $\exists w \in \wv(P). \; P_w \geq 0$ and $\exists w \in \wv(Q). \; Q_w \geq 0$.

\subsection{Technical Lemmas}

These lemmas will be used throughout our proof.

\begin{lemma}[Non-Negative Potential] \label{lem:nonneg} If $P \geq 0$, then $\Phi(V: \langle \Gamma; P\rangle) \geq 0$.
\end{lemma}

\begin{proof}
Potential is calculated as a sum of non-negative values scaled by coefficients in $P$.
\end{proof}

\begin{lemma}[Subtype Potential] \label{lem:sub} If $\Gamma \vdash P \sub Q $, then $\Phi(V: \langle \Gamma; P\rangle) \geq \Phi(V: \langle \Gamma; Q\rangle)$
\end{lemma}

\begin{proof}
The subtyping rules require potential to be pointwise less, except in the case of function types. However, because function types always carry 0 potential, this final case holds.
\end{proof}

\subsection{Relax}

Suppose this was the final type rule applied:
\begin{mathpar}
\inferrule[T-Relax]{
	\Gamma \mid R \vdash e: \tau \mid  S
	\\
	\Gamma \vdash P \sub R
	\\
	\Gamma, \circ: \tau \vdash S \sub Q 
}{
	\Gamma \mid P  \vdash e: \tau \mid Q
}
\end{mathpar}
Then we may assume that its premisses hold. Because $\Gamma \mid R \vdash e: \tau \mid  S$ involves the same expression $e$ as the conclusion, we can use the assumed cost derivation for the conclusion to apply the inductive hypothesis. We then learn the following, where $\wv(R)=\wv(P)$ and $\wv(S)=\wv(Q)$.
\begin{align}
\exists w \in \wv(R). \; & p \leq \Phi(V : \langle \Gamma; R_w \rangle) \label{exrelax}
\\
\forall u \in \wv(S). \; \exists w \in \wv(R). \; & p-q \leq \Phi(V : \langle \Gamma; R_w \rangle) - \Phi(V,v : \langle \Gamma,\circ:\tau; S_u \rangle ) \label{unrelax}
\end{align} 
Combining \cref{lem:sub} with the remaining premisses, we find that \cref{ex} follows directly from \cref{exrelax}, and \cref{un} follows directly from \cref{unrelax}. This completes the relax case.

\subsection{Superposition}
Suppose any of these was the final type rule applied:
\begin{mathpar}
     \inferrule[T-Superposition-In]{
	\Gamma \mid P, u \mapsto P_w  \vdash e : \tau \mid Q
}{
	\Gamma \mid P \vdash e : \tau \mid Q
}

     \inferrule[T-Superposition-Out]{
	\Gamma \mid P \vdash e: \tau \mid Q
}{
	\Gamma \mid P \vdash e: \tau \mid Q, u \mapsto Q_w
}

     \inferrule[T-Collapse-In]{
	\Gamma \mid P \vdash e  : \tau \mid Q
}{
	\Gamma \mid P, w \mapsto R \vdash e  : \tau \mid Q
}

     \inferrule[T-Collapse-Out]{
	\Gamma \mid P \vdash e:\tau \mid Q, w \mapsto R
}{
	\Gamma \mid P \vdash e:\tau \mid Q
}
\end{mathpar}

Then we may assume its premiss holds. Because each premiss involves the same expression $e$ as the conclusion, we can use the assumed cost derivation for the conclusion to apply the inductive hypothesis.  

No matter the particular rule, the set of $\{P_w'\}$ from premiss context $P'$ is a subset of that of the conclusion. Thus, the existential witness of \cref{ex} from the premiss's inductive hypothesis still holds for the identical expression $e$ in the conclusion. Thus, the conclusion satisfies \cref{ex}.

Similarly, the set of $\{Q_w'\}$ from the premiss remainder is a superset of that of the conclusion. Because the premiss's inductive hypothesis \cref{un} already satisfies the universal quantification over the larger set, \cref{un} continues to hold for the identical expression $e$ in the conclusion. Thus, the conclusion also satisfies \cref{un}.

\subsection{Var}

Suppose this was the final type rule applied:
\begin{mathpar}
\inferrule[T-Var]{
	\pi_x(P) = \pi_x(Q) + \pi_\circ(Q)
	\\
	\pi_{\neg x} (P) =  \pi_{\neg \{x,\circ\}} (Q) 
}{
	\Gamma, x:\tau \mid P \vdash x: \tau \mid Q
}
\end{mathpar}
Because $Q$'s worldviews are pointwise related to $P$'s, we may assume without loss of generality that $\wv(P)=\wv(Q)$, and that their labelling is consistent with the pointwise behaviour.

There is only one compatible cost rule that could end the cost derivation, and we may assume its premisses hold as well: 
\begin{mathpar}
\infer[\!\!\mathit{Var}]{
        V \vdash x \Downarrow v \mid (0,0)
    }{
    	V(x) =v
    }
\end{mathpar}
Thus we want to prove specifically that
\begin{align}
    \exists w \in \wv(P). \; & 0 \leq \Phi(V : \langle \Gamma; P_w \rangle) \label{exvar}
    \\
    \forall u \in \wv(Q). \; \exists w \in \wv(P). \; & 0 \leq \Phi(V : \langle \Gamma; P_w \rangle) - \Phi(V,v : \langle \Gamma,\circ:\tau; Q_u \rangle ) \label{unvar}
\end{align}
\cref{exvar} follows immediately from the side condition combined with \cref{lem:nonneg}. Then we finish off the case by finding that the premisses express a conservation of potential $\Phi(V : \langle \Gamma; P_w \rangle) = \Phi(V,v : \langle \Gamma,\circ:\tau; Q_w \rangle )$, so that \cref{unvar} is satisfied with an equality at the same worldview for $P$ and $Q$.

\subsection{Tick}
Suppose this was the final type rule applied:
\begin{mathpar}
\inferrule[T-Tick]{
	P(*) - r = Q(*)
	\\
	\pi_{\neg *} (P) =  \pi_{\neg \{*,\circ\}} (Q) 
}{
	\Gamma \mid P \vdash \tick r : \unit \mid Q
}
\end{mathpar}
Because $Q$'s worldviews are pointwise related to $P$'s, we may assume without loss of generality that $\wv(P)=\wv(Q)$, and that their labelling is consistent with the pointwise behaviour.

There is only one compatible cost rule that could end the cost derivation, and we may assume its premisses hold as well:
\begin{mathpar}
\infer[\mathit{Tick}]{
		V \vdash \mathit{tick}\{r\} \Downarrow () \mid (p,q) 
	}{
		p = max(r,0)
		&
		q = max(-r,0)
	}
\end{mathpar}
Thus, after applying some algebra, we want to prove specifically that:
\begin{align}
    \exists w \in \wv(P). \; & max(r,0) \leq \Phi(V : \langle \Gamma; P_w \rangle) \label{extick}
    \\
    \forall u \in \wv(Q). \; \exists w \in \wv(P). \; & r \leq \Phi(V : \langle \Gamma; P_w \rangle) - \Phi(V, () : \langle \Gamma,\circ:\unit; Q_u \rangle ) \label{untick}
\end{align}
When $r$ is negative, \cref{extick} follows because of \cref{lem:nonneg} and the side condition on $P$. When $r$ is positive, \cref{extick} follows because of \cref{lem:nonneg}, the side condition on $Q$, and the premisses that ensure that the context has exactly $r$ more potential than the remainder, $\Phi(V : \langle \Gamma; P_w \rangle) = \Phi(V,v : \langle \Gamma,\circ:\unit; Q_w \rangle ) + r$. This same equation also ensures \cref{untick} is satisfied with an equality at the same worldview for $P$ and $Q$..

\subsection{Let}
Suppose this was the final type rule applied:
\begin{mathpar}
     \inferrule[T-Let]{
	\Gamma \mid P \vdash e_1 : \tau \mid R
	\\
	\Gamma, x:\tau \mid R[x/\circ] \vdash e_2 : \rho \mid Q
	\\
	\pi_x(Q) \geq 0
    }{
    	\Gamma \mid P \vdash \elet x {e_1} {e_2} : \rho \mid \pi_{\neg x}(Q)
    }
\end{mathpar}
Then we may assume its premisses hold. Further, there is only one compatible cost rule that could end the cost derivation, and we may assume its premisses hold as well:
\begin{mathpar}
     \infer[\mathit{Let}]{
        V \vdash \elet x {e_1} {e_2} \Downarrow v_2 \mid (p+max(p'-q,0),q'+max(q-p',0))
    }{
        V \vdash e_1 \Downarrow v_1 \mid (p,q)
        & 
        V[x \mapsto v_1] \vdash e_2 \Downarrow v_2 \mid (p',q')
    }
\end{mathpar}
Thus, after applying some algebra, we want to prove specifically that:
\begin{align}
    \exists w \in \wv(P). \; & p + max(p'-q,0) \leq \Phi(V : \langle \Gamma; P_w \rangle) \label{exlet}
    \\
    \forall u \in \wv(Q). \; \exists w \in \wv(P). \; & p+p'-q-q' \leq \Phi(V : \langle \Gamma; P_w \rangle) \nonumber
    \\
    &- \Phi(V,v_2 : \langle \Gamma,\circ:\tau; Q_u \rangle ) \label{unlet}
\end{align}
The premisses of each let rule allow us to apply our inductive hypothesis to find that:
\begin{align}
     \exists w \in \wv(P). \; & p  \leq \Phi(V : \langle \Gamma; P_w \rangle) \label{exlet1}
    \\
    \forall u \in \wv(R). \; \exists w \in \wv(P). \; & p-q \leq \Phi(V : \langle \Gamma; P_w \rangle) - \Phi(V,v_1 : \langle \Gamma,\circ:\tau; R_u \rangle ) \label{unlet1}
    \\
    \exists w \in \wv(R). \; & p' \leq \Phi(V[x\mapsto v_1] : \langle \Gamma,x:\tau ; R_w[x/\circ] \rangle) \label{exlet2}
    \\
    \forall u \in \wv(Q). \; \exists w \in \wv(R). \; & p'-q' \leq \Phi(V[x\mapsto v_1] : \langle \Gamma,x:\tau; R_w[x/\circ] \rangle) \nonumber
    \\
    &- \Phi(V[x\mapsto v_1],v_2 : \langle \Gamma,,x:\tau,\circ:\rho; Q_u \rangle ) \label{unlet2}
\end{align}

If $p' \geq q$, then \cref{exlet} can be derived by using the existential witness of \cref{exlet2}, plugging it into the universal quantification of \cref{unlet1}. Then add the inequalities at these witnesses from \cref{unlet1} and \cref{exlet2}, and cancel equivalent values, yielding \cref{exlet}. Otherwise, if $p' < q$, then \cref{exlet1} satisfies \cref{exlet}.

To derive \cref{unlet}, one considers an arbitrary $u \in \wv(Q)$, and finds the existential witness $w \in \wv(R)$ guaranteed by \cref{unlet2}. This $w$ is then plugged into the universal quantification of \cref{unlet1} to get a $w' \in \wv(P)$. Finally, add the remaining inequalities of \cref{unlet1} and \cref{unlet2} at $w'$ and $u$ respectively, and cancel equivalent potential values. The resulting inequality can be weakened using \cref{lem:nonneg} and the premiss that $\pi_x(Q) \geq 0$, yielding \cref{unlet}.

\subsection{Cond}
Suppose this was the final type rule applied:
\begin{mathpar}
\inferrule[T-Cond]{
	\Gamma, b:\tbool \mid P \vdash e_1 : \tau \mid Q
	\\
	\Gamma, b:\tbool \mid P \vdash e_2 : \tau \mid Q
}{
	\Gamma, b:\tbool \mid P \vdash \ite b {e_1} {e_2} : \tau \mid Q
}
\end{mathpar}
Then we may assume its premisses hold. Further, there are only two compatible cost rules that could end the cost derivation, and we may the premisses hold from one of them:
\begin{mathpar}
\infer[\mathit{CondT}]{
            V \vdash \ite {x_b} {e_t} {e_f} \Downarrow v \mid (p,q)
        }{
           	V(x_b) = \mathit{true}
		& 
		V \vdash e_t \Downarrow v \mid (p,q)
        }

        \infer[\mathit{CondF}]{
            V \vdash \ite {x_b} {e_t} {e_f} \Downarrow v \mid (p,q)
        }{
           	V(x_b) = \mathit{false}
		& 
		V \vdash e_f \Downarrow v \mid (p,q)
	}
\end{mathpar}
In either case, note that the cost and contexts are unchanging between the whole conditional and the single branch. Thus the inductive hypothesis gained from the typing and cost rules' premisses satisfies that needed for the conclusion.   

\subsection{Fun-Rec}
Suppose this was the final type rule applied:
\begin{mathpar}
 \inferrule[T-Fun-Rec]{
	\Gamma, f: \tau \rightarrow \rho, x:\tau \mid Q \vdash e : \rho \mid R
	\\
	\lfloor \pi_{\neg \{x,*\}}(Q) \rfloor
	\\
	\pi_{\neg\{x,*\}}(Q) = \pi_{\neg \{x,*,\circ\}}(R)
	\\
	\pi_x(Q)[* \mapsto Q(*)] = \pi_{f \mathsf a}(Q)
	\\
	\pi_\circ(R) = \pi_{f\mathsf b}(Q)
	\\
	\pi_x(R)[ * \mapsto R(*)] = \pi_{f\mathsf c}(Q)
	\\
	\pi_\circ(P) = \pi_f(Q)
}{
	\Gamma \mid \pi_{\neg \circ} (P) \vdash \fun f x e : \tau \rightarrow \rho \mid P
}
\end{mathpar}
There is only compatible cost rule that could end the cost derivation:
\begin{mathpar}
 \infer[\mathit{Fun}]{
        		V \vdash \fun f x e \Downarrow \closure f V x e \mid (0,0)
        }{
        }
\end{mathpar}
Thus, after applying some algebra, we want to prove specifically that:
\begin{align}
    \exists w \in \wv(P). \; & 0 \leq \Phi(V : \langle \Gamma; \pi_{\neg \circ}(P_w) \rangle) \label{exfun}
    \\
    \forall u \in \wv(P). \; \exists w \in \wv(P). \; & 0 \leq \Phi(V : \langle \Gamma; \pi_{\neg \circ}(P_w) \rangle) \nonumber
    \\
    & - \Phi(V,\closure f V x e : \langle \Gamma,\circ:\tau \rightarrow \rho; P_u \rangle ) \label{unfun}
\end{align}

\cref{exfun} follows immediately from the side condition combined with \cref{lem:nonneg}. \cref{unfun} is then satisfied by equality because functions carry no potential, and otherwise both potential terms are identical.

\subsection{App}
Suppose this was the final type rule applied:
\begin{mathpar}
    \inferrule[T-App]{
    	\exists w \in \wv(P). \; \tau \vdash \pi_x(P_w)[* \mapsto P_w(*)] \sub \pi_{f\mathsf a}(P_w) \wedge P_w \geq 0
    	\\
    	\pi_{x}(P)[x*\mapsto P(*)] - \pi_{f\mathsf a}(P) + \pi_{f\mathsf c}(P) = \pi_{x}(Q)[x*\mapsto Q(*)]
    	\\
    	\pi_{f\mathsf b}(P) = \pi_\circ(Q)
    	\\
    	\pi_{\neg x}(P) = \pi_{\neg\{x,\circ\}}(Q)
    }{
    	\Gamma, f:\tau \rightarrow \rho, x:\tau \mid P  \vdash f\;x : \rho \mid Q
    }
\end{mathpar}
Because $Q$'s worldviews are pointwise related to $P$'s, we may assume without loss of generality that $\wv(P)=\wv(Q)$, and that their labelling is consistent with the pointwise behaviour.

There is only compatible cost rule that could end the cost derivation:
\begin{mathpar}
\infer[\mathit{App}]{
            V \vdash f\; x \Downarrow v  \mid (p,q)
        }{
        	    V(x) = v_x
	    &
            V(f) = \closure g {V'} {x'} {e}
	    &
            V'[x' \mapsto v_x, g \mapsto \closure g {V'} {x'} {e}] \vdash e \Downarrow v \mid (p,q)
        }
\end{mathpar}
Thus, after applying some algebra, we want to prove specifically that:
\begin{align}
    \exists w \in \wv(P). \; & p \leq \Phi(V : \langle \Gamma, f:\tau \rightarrow \rho, x:\tau; P_w \rangle) \label{exapp}
    \\
    \forall u \in \wv(Q). \; \exists w \in \wv(P). \; & p-q \leq \Phi(V : \langle \Gamma, f:\tau \rightarrow \rho, x:\tau; P_w \rangle) \nonumber
    \\
    &- \Phi(V,v : \langle \Gamma, f:\tau \rightarrow \rho, x:\tau,\circ:\rho; Q_u \rangle ) \label{unapp}
\end{align}
Further, due to the premiss $V:\langle \Gamma; P \rangle$, we know that the value of $f$ must be typable as $g:\langle \tau \rightarrow \rho ; \pi_f(P_{w'})\rangle$ for each $w' \in \wv(P)$. Thus, the following typing rule must be applicable, where $V': \langle \Gamma'; P' \rangle$ and $|\wv(P')|=1$:
\begin{mathpar}
 \inferrule[T-Fun-Rec]{
	\Gamma', g: \tau \rightarrow \rho, x':\tau \mid Q' \vdash e : \rho \mid R'
	\\
	\lfloor \pi_{\neg \{x',*\}}(Q') \rfloor
	\\
	\pi_{\neg\{x',*\}}(Q') = \pi_{\neg \{x',*,\circ\}}(R')
	\\
	\pi_{x'}(Q')[* \mapsto Q'(*)] = \pi_{g \mathsf a}(Q')
	\\
	\pi_\circ(R') = \pi_{f\mathsf b}(Q')
	\\
	\pi_{x'}(R')[ * \mapsto R'(*)] = \pi_{g\mathsf c}(Q')
	\\
	\pi_\circ(P') = \pi_f(Q')
}{
	\Gamma' \mid \pi_{\neg \circ} (P') \vdash \fun g {x'} {e} : \tau \rightarrow \rho \mid P'
}
\end{mathpar}
We may therefore assume the premisses of this rules hold as well. Combining the typing judgment premiss of this rule with the cost rule's premiss allows us to apply the inductive hypothesis, learning the following inequalities. Note that the quantification is trivial, because each has only one worldview.
\begin{align}
 \exists w \in \wv(Q'). \; & p \leq \Phi(V' : \langle \Gamma, g:\tau \rightarrow \rho, x':\tau; Q_w' \rangle) \label{exapp1}
    \\
    \forall u \in \wv(Q'). \; \exists w \in \wv(R'). \; & p-q \leq \Phi(V' : \langle \Gamma, g:\tau \rightarrow \rho, x':\tau; Q_w' \rangle) \nonumber
    \\
    &- \Phi(V',v : \langle \Gamma, g:\tau \rightarrow \rho, x':\tau,\circ:\rho; R_u' \rangle ) \label{unapp1}
\end{align}
To derive \cref{exapp}, first recall that the potential assigned by $P$ to the function's argument type is the same in every worldview, so we need only consider 1 case of how $Q'$ assigns potential. Some of the remaining premisses of the latter typing rule ensure that $Q'$ imbues zero potential aside from that on $x'$ and the ambient constant potential. There is a premiss on the former typing rule that ensures that there is some worldview in which $P$ imbues at least as much potential. At the witnessing world, \cref{lem:sub} ensures that $P$ gives $x$ and the constant potential at least as much as $Q'$ gave $x'$ and its constant potential, and otherwise $P$ supplies a non-negative amount of potential where $Q'$ supplies zero. These facts combined with \cref{exapp1} yield \cref{exapp}.

To derive \cref{unapp}, consider the particular witness typing for $P$'s worldview $w'$. We first use some of the remaining premisses of the latter typing rule to find that $R'$ assigns no potential other than on the argument $x'$, the ambient constant potential, and the return $\circ$. The last premisses on the latter typing rule ensure that the potential annotations ensure that the potential annotation of these possibly-nonzero values is assigned to the return and remainder on $g$. Thus, the difference in potential expressed in \cref{unapp1} is exactly the difference expressed between the argument and return-plus-remainder of the function $f$ in $P_w'$. In turn, the remaining premisses of the former typing rule ensure that the difference in potential on the function type in $P_{w'}$ is exactly the difference in potential expressed in \cref{unapp} when at worldview $w'$ for both $P$ and $Q$. Thus, we find the inequality is met with the same worldview for both $P$ and $Q$ in every case, and \cref{unapp} is satisfied.

\subsection{Leaf}
Suppose this was the final type rule applied:
\begin{mathpar}
\inferrule[T-Leaf]{
}{
	\Gamma \mid P \vdash \leaf : T(\tau)  \mid P
}
\end{mathpar}
There is only compatible cost rule that could end the cost derivation:
\begin{mathpar}
 \infer[\mathit{Leaf}]{
            V \vdash \leaf \Downarrow \leaf \mid (0,0)
        }{
        }
\end{mathpar}
Thus, after applying some algebra, we want to prove specifically that:
\begin{align}
    \exists w \in \wv(P). \; & 0 \leq \Phi(V : \langle \Gamma; P_w \rangle) \label{exleaf}
    \\
    \forall u \in \wv(P). \; \exists w \in \wv(P). \; & 0 \leq \Phi(V : \langle \Gamma; P_w \rangle) - \Phi(V,v : \langle \Gamma,\circ:T(\tau); P_u \rangle ) \label{unleaf}
\end{align}
\cref{exleaf} follows immediately from the side condition combined with \cref{lem:nonneg}. Then, because leaves carry no potential, \cref{unleaf} is always satisfied with an equality at the same worldview for $P$ and $Q$.
    
\subsection{Node}

Suppose this was the final type rule applied:
\begin{mathpar}
 \inferrule[T-Node]{
	P \RHD_{s,t_1,t_2,\circ} Q
}{
	\Gamma, s:\tau, t_1 : T(\tau), t_2 : T(\tau) \mid P \vdash \node {t_1} s {t_2} : T(\tau) \mid Q
}
\end{mathpar}
There is only compatible cost rule that could end the cost derivation:
\begin{mathpar}
 \infer[\mathit{Node}]{
           V \vdash \node {t_1} {a} {t_2} \Downarrow \node {v_1} {v_a} {v_2}\mid  (0,0)     
        }{
        	V(t_1) = v_1
	&
	V(t_2) = v_2
	& 
	V(a) = v_a
        }
\end{mathpar}
Thus, after applying some algebra, we want to prove specifically that:
\begin{small}
\begin{align}
    \exists w \in \wv(P). \; & 0 \leq \Phi(V,v_a, v_1,v_2 : \langle \Gamma, s:\tau, t_1 : T(\tau), t_2 : T(\tau); P_w \rangle) \label{exnode}
    \\
    \forall u \in \wv(Q). \; \exists w \in \wv(P). \; & 0 \leq \Phi(V,v_a, v_1,v_2 : \langle \Gamma, s:\tau, t_1 : T(\tau), t_2 : T(\tau); P_w \rangle) \nonumber
    \\
    &- \Phi(V,v_a, v_1,v_2, \node {v_1} {v_a} {v_2} : \langle \Gamma, s:\tau, t_1 : T(\tau), t_2 : T(\tau),\circ:T(\tau); Q_u \rangle ) \label{unnode}
\end{align}
\end{small}
\cref{exnode} is satisfied by combining \cref{lem:nonneg} with the side condition.
Then to derive \cref{unnode}, we first unpack the premiss of the type rule:

\begin{scriptsize}
\begin{align*}
    P \RHD_{s,t_1,t_2,\circ} Q \equiv\;  & \forall u \in \wv(Q). \; \exists v,w \in \wv(P).\; \pi_{\neg\{s,t_1,t_2,\circ,*\}}(P_w) = \pi_{\neg\{s,t_1,t_2,\circ,*\}}(P_v) = \pi_{\neg \{s,t_1,t_2,\circ, *\}}(Q_u)
\\
& \wedge R_v = \pi_{\{s,t_1,t_2,*\}}(P_v) - \pi_{\{s,t_1,t_2,*\}}(Q_u) \wedge R_w = \pi_{\{s,t_1,t_2,*\}}(P_w) - \pi_{\{s,t_1,t_2,*\}}(Q_u) 
\\
& \wedge \pi_{s}(R_v)[*\mapsto Q_u(\circ\mathsf e*)] = \pi_{t_1 \mathsf e}(R_v) = \pi_{t_2 \mathsf e}(R_v)=\pi_{s}(R_w)[*\mapsto Q_u(\circ \mathsf e*)]  = \pi_{t_1 \mathsf e}(R_w) = \pi_{t_2 \mathsf e}(R_w) = \pi_{\circ \mathsf e} (Q_u)
\\
& \wedge R_v(t_1\mathsf d)=  R_w(t_2\mathsf d) =Q_u(\circ \mathsf d) \wedge R_v(t_2\mathsf d)=  R_w(t_1\mathsf d) = 0
\\
&\wedge R_v(*) = R_w(*) = Q_u(\circ \mathsf d) + Q_u(\circ \mathsf e *)
\end{align*}
\end{scriptsize}

The first line of this matches the quantification we need, picking two corresponding worldviews in $P$ for each one in $Q$. This line then tells us that there is always a worldview of $P$ with equal context annotations to that of $Q$, so these cancel and we need only consider what happens to the values being directly operated on ($v_a,v_1,v_2$) and the ambient potential. The second line begins defining an auxiliary potential assignment $R$ that tracks the difference in potential annotations between each worldview in $Q$ and and its corresponding pair of worldviews in $P$. The third line assures us that all the node- and payload-based potential is identical between that taken from the values operated on and the resulting tree $ \node {v_1} {v_a} {v_2}$. Because this resulting tree combines the nodes of the values operated on, these kinds of potentials will also cancel, save for 1 unit of per-node potential at the root. The fourth line tells us that in worldview $v$, the depth potential of the output tree is equal to that of $t_1$ while $t_2$ has no depth potential, and meanwhile in worldview $w$ the two subtrees swap these roles. Since one of these subtrees actually witnesses the maximum depth the world where this occurs accounts for all but the root node of $\circ$'s depth potential. The final line ensures that the ambient potential pays for the remaining depth- and node-based potential at the root.

Thus, \cref{unnode} is satisfied for an arbitrary worldview $u \in \wv(Q)$ with an equality at one of the worldviews $v,w \in \wv(P)$. That worldview is the one putting all depth-potential on the subtree of maximum depth, and otherwise using potential annotations that cancel perfectly with those in $Q_u$.

\subsection{Match Tree}

Suppose this was the final type rule applied:
\begin{mathpar}
\inferrule[T-Match-Tree]{
	\Gamma,t: T(\tau) \mid P \vdash e_1: \rho \mid Q
	\\
	\Gamma,t: T(\tau),  s:\tau, t_1 : T(\tau), t_2 : T(\tau) \mid R \vdash e_2: \rho \mid S
	\\
	P \lhd_{s,t_1,t_2,t} R
	\\
	S' \RHD_{s,t_1,t_2,t} Q'
	\\
	\pi_{\neg t}(S') =\pi_{\neg t}(S) 
	\\
	\pi_{\neg t}(Q') =\pi_{\neg t}(Q) 
	\\
	\pi_t(S)-\pi_t(S') = \pi_t(Q)-\pi_t(Q') = \pi_t(R)
}{
	\Gamma, t: T(\tau) \mid P \vdash \matcht t {e_1} {t_1} s {t_2} {e_2} :\rho \mid Q
}
\end{mathpar}
There are only two compatible cost rules that could end the cost derivation:
\begin{mathpar}
 \infer[\!\!\mathit{TMatchL}]{
           V \vdash  \matcht x {e_1} {t_1} a {t_2} {e_2} \Downarrow v \mid  (p,q)    
        }{
            V(x) = \leaf
            &
            V \vdash e_1 \Downarrow v \mid (p,q)
        }
        
 \infer[\!\!\mathit{TMatchT}]{
            V \vdash  \matcht x {e_1} {t_1} a {t_2} {e_2} \Downarrow v \mid  (p,q)     
        }{
            V(x) = \node {v_1} {v_a} {v_2} 
            &
            V[t_1 \mapsto v_1, t_2 \mapsto v_2, a \mapsto v_a] \vdash e_2 \Downarrow v \mid (p,q)
        }
\end{mathpar}

Consider the first cost rule for leaves. Across the typing and cost rule, the values of $P,Q,V,v,p,q$ are identical between the conclusions and premisses (the first premiss specifically in the case of the typing rule). Thus the inductive hypothesis gives exactly the properties we desire.

Now consider the second cost rule for nodes. In this case we would like to prove the following.

\begin{align}
    \exists w \in \wv(P). \; & p \leq \Phi(V : \langle \Gamma, x : T(\tau); P_w \rangle) \label{exmatcht}
    \\
    \forall u \in \wv(Q). \; \exists w \in \wv(P). \; & p-q \leq \Phi(V : \langle \Gamma, x : T(\tau); P_w \rangle) \nonumber
    \\
    &- \Phi(V,v : \langle \Gamma,x: T(\tau),\circ:\rho; Q_u \rangle ) \label{unmatcht}
\end{align}
We can combine the second premiss of the type rule with the premiss of the second cost rule to apply our inductive hypothesis and get the following, where $V' = V[t_1 \mapsto v_1, t_2 \mapsto v_2, a \mapsto v_a]$:
\begin{align}
     \exists w \in \wv(R). \; & p \leq \Phi(V': \langle \Gamma, x : T(\tau), t_1:  T(\tau), t_2 : T(\tau), a: T(\tau); R_w \rangle) \label{exmatchtp}
    \\
    \forall u \in \wv(S). \; \exists w \in \wv(R). \; & p-q \leq \Phi(V' : \langle \Gamma, x : T(\tau), t_1:  T(\tau), t_2 : T(\tau), a: T(\tau); R_w \rangle) \nonumber
    \\
    &- \Phi(V',v : \langle \Gamma, x : T(\tau), t_1:  T(\tau), t_2 : T(\tau), a: T(\tau), \circ:\rho ; S_w \rangle) \label{unmatchtp}
\end{align}
Now let us determine how $P$ relates to $R$. This is given by the following:
\begin{align*}
    P \lhd_{s,t_1,t_2,t_3} R \equiv \; & \pi_{\neg\{s,t_1,t_2,t_3,*\}}(P) = \pi_{\neg\{s,t_1,t_2,t_3,*\}}(R) \wedge R'=\pi_{t_3}(P) - \pi_{t_3}(R) 
\\
&\wedge \pi_s(R)[* \mapsto R'(e*)] = \pi_{t_1\mathsf e}(R) =\pi_{t_2\mathsf e}(R) = \pi_e(R')
\\
& \wedge R'(\mathsf d) = \pi_{t_1\mathsf d}(R) + \pi_{t_2\mathsf d}(R) \wedge R(*) = P(*) + R'(\mathsf e *) + R'(\mathsf d)
\end{align*}
The first line begins by establishing that $P$ and $R$ are identical except the values being operated on ($v_a,v_1,v_2, \node {v_1}{v_a}{v_2}$) and the ambient potential. It then defines an auxiliary potential assignment $R'$ tracking the amount of potential taken from $t$ in $P$ for use in $R$. The second line then establishes that the node- and payload-based potential is identical between that taken from $t$ and that assigned to the types of $t_1$ and $t_2$, as well as the type of $a$. Because the nodes of $v_a,v_1,v_2$ together make up $t: \node {v_1}{v_a}{v_2}$, so far all potential cancels (except for the extra per-node potential on $t$'s root). The final line tells us that the per-depth potential taken from $t$ is split convexly between its subtrees, which can account for at most all the depth potential on $t$ save for that on the root node in the case that the convext split assigns all potential to the deepest subtree. The final line also that the ambient potential is adjusted to account for the node- and depth- based potential remaining from that taken off $t$'s root node. 

To summarize, the potential between $P$ and $R$ cancels at all points except the convex split into the subtrees $v_1,v_2$ where $R$ might have less. Therefore, $\Phi(V: \langle \Gamma,x:T(\tau); P_w\rangle) \geq \Phi(V' : \langle \Gamma, x:T(\tau), s: \tau, t_1 : T(\tau), t_2 : T(\tau); R_w \rangle)$ in all worldviews $w\in \wv(P)$. This gives us \cref{exmatcht} from \cref{exmatchtp}.

Now we determine how $Q$ relates to $S$. The premisses in the typing rule first define $Q'$ and $S'$, which are identical to $Q$ and $S$ respectively except at $t$, where they differ by the amount of potential left on $t$ in $R$, which is the amount not used in the pattern match. $S'$ and $Q'$ are then related as if they were to construct $t$ from $s,t_1,t_2$. From the node case of our induction, we know this guarantees that for each worldview $u \in \wv(Q')$, there exists one in $S'$ assigning the same total potential across their respective values. Because $Q$ and $S$ differ by the same amount of potential from $Q'$ and $S'$, the same relation holds for them. 

We can now use this witnessing worldview in $S$ with \cref{unmatchtp} to get a witnessing worldview in $w \in \wv(R)$ assigning total potential no less than $p-q$ more than $Q_u$ across their respective values. That is:
\begin{align*}
    \forall u \in \wv(Q). \; \exists w \in \wv(R). \; & p-q \leq \Phi(V' : \langle \Gamma, x : T(\tau), t_1:  T(\tau), t_2 : T(\tau), a: T(\tau); R_w \rangle) \nonumber
    \\
    &- \Phi(V,v : \langle \Gamma,x: T(\tau),\circ:\rho; Q_u \rangle )
\end{align*}
Weakening this inequality by the relation we have already determined between $P$ and $R$ yields \cref{unmatcht}, completing this case.

\end{document}